%% file: cdc_arxive.tex
\documentclass[letterpaper, 10pt, conference]{ieeeconf}

\IEEEoverridecommandlockouts
\usepackage{graphicx}
\usepackage{hyperref}
\usepackage{amsmath}
\usepackage{amssymb}
\usepackage{xcolor}
\usepackage{comment}
\usepackage{pgfplots}
\usepackage{mathtools}
\usepackage{subcaption}
\pgfplotsset{compat=1.18}

\usepackage{algorithm}
\usepackage{algpseudocode}

\newtheorem{theorem}{Theorem}[section]
\newtheorem{lemma}[theorem]{Lemma}
\newtheorem{proposition}[theorem]{Proposition}
\newtheorem{corollary}[theorem]{Corollary}

\newtheorem{definition}[theorem]{Definition}
\newtheorem{remark}{Remark}[section]

\newtheorem{problem}{Problem}

\input{macros-abbrev.tex}
\newcommand{\Ptwo}{\ensuremath{\mathcal{P}_2(\Omega)}}
\newcommand{\PtwoM}{\ensuremath{\mathcal{P}_2(M)}}
\newcommand{\PM}{\ensuremath{\mathcal{P}_M}}
\newcommand{\diff}{\ensuremath{\mathrm{d}}}

\newcommand{\Id}{\operatorname{Id}}

\renewcommand{\dist}{\operatorname{dist}}
\newcommand{\bulletsym}{\ensuremath{\bullet}}
\newcommand{\bulletend}{%
  \relax
  \ifmmode
  \else
    \unskip\hfill
  \fi
  \bulletsym
}

\newcommand{\oprocendsymbol}{\hbox{$\bullet$}} % \diamond means proof not included
\newcommand{\oprocend}{\relax\ifmmode\else\unskip\hfill\fi\oprocendsymbol}

\newcommand{\longthmtitle}[1]{\mbox{}{\textit{(#1):}}}

\makeatother

\title{\LARGE \bf Input-to-State Stability of Gradient Flows in
  Distributional Space}

\author{Guillem Pascual and Sonia Mart{\'\i}nez\thanks{G. Pascual and
    S. Mart{\'\i}nez are with the Department of Mechanical and
    Aerospace Engineering at the University of California, San Diego,
    CA, USA. (e-mail: \{gpascualastivill, soniamd\}@ucsd.edu.)}}

\begin{document}

\maketitle

\begin{abstract}
  This paper proposes a new notion of \textit{distributional
    Input-to-State Stability} (dISS) for dynamic systems evolving in
  probability spaces over a domain. Unlike other norm-based ISS
  concepts, we rely on the Wasserstein metric, which captures more
  precisely the effects of the disturbances on atomic and non-atomic
  measures.  We show how dISS unifies both ISS and Noise to State
  Stability (NSS) over compact domains for particle dynamics, while
  extending the classical notions to sets of probability
  distributions. We then apply the dISS framework to study the
  robustness of various Wasserstein gradient flows with respect to
  perturbations. In particular, we establish dISS for gradient flows
  defined by a class of $l$-smooth and $\lambda$-convex 
  functionals subject to bounded
  disturbances, such as those induced by entropy in optimal transport.
  Further, we study the dISS robustness of the large-scale algorithms
  when using Kernel and sample-based approximations. This results into
  a characterization of the error incurred when using a finite number
  of agents, which can guide the selection of the swarm size to
  achieve a mean-field objective with prescribed accuracy and
  stability guarantees.

  % While dISS unifies and
  % is equivalent to ISS (resp.~Noise ISS), of deterministic
  % (resp. stochastic) particle dynamics over compact domains, the
  % former is a stronger notion over non-compact ones. Then, we apply
  % the dISS framework to study disturbances of transport algorithms of
  % large-scale multi-agent systems. In particular, we pay attention to
  % the case of stochastically perturbed and limited-range versions of
  % continuous and discrete-time optimal transport algorithms.
\end{abstract}

%% Options for titles (10 words is the best max): 1) Distributional
%% ISS (too brief), 2) Distributional ISS of Very Large Swarms, 3)
%% include the wording 'optimal transport' in the title
%% e.g. input-to-state stability in distributional space with
%% applications to optimal transport

%%\margin{ In a journal version, we should definite write Dynamic
%%  Systems. If we consider a consensus problem on Barycenters as
%%  Carrillo and student, we would have a system that is not gradient
%%  and still we can study its dISS properties---it would be a nice
%%  picture}

\section{Introduction}

Recent advances in artificial intelligence have renewed the interest
in the practical deployment of very large robotic swarms. % ; with
% applications in disaster response, environmental monitoring to
% agriculture.
Made of low-cost simple robots, swarms are highly adaptive while
exhibiting remarkable resilience. %through their large numbers.
Yet, their coordinated control remains complex, as is affected by
multiple
%disturbances arising from
communication and sensor failures, enviromental events, and
uncertainty. Thus, analyzing the impact of the effect of disturbances
% goals  of the whole ensamble of agents
becomes paramount to ensuring swarm safety and graceful stable
degradation. Particularly motivated by this, we seek to develop and
analyze a bounded input-disturbance-to-state stability notion that is
applicable to general probability measures.  %that model
%swarms.
This question is also relevant in stochastic control and machine
learning applications, such as generative learning, which rely on
algorithms that evolve in spaces of probabilities. 
%  \new{The efficient
%  approximation of such dynamic systems can be understood as a type of
%  perturbation, thus the accurate characterization of their effects
%  becomes important.}

% The state of swarm systems is often modeled as a probability
% distribution, representing the density of agents in the domain.  Thus,
% our goal can be summarized as analysing the effect of disturbances on
% dynamic systems defined in probability spaces.

% \margin{We could motivate the research not only by swarms, but also by
%   interest in generative learning in ML, and stochastic control
%   systems. Recall the introduction that Panos Tsiotras gave in his
%   seminar. I've added this paragraph, but we could phrase the whole
%   opening motivating paragraph in this more general (and ambitious)
%   tone. If we do that, the paper could be read by more people}

%This is motivated by the coordination of large-scale
%multi-agent systems, where the state of the system is modeled via 
%probability distributions which encode the density of agents in the domain.
%The disturbances that agents are subject to perturb the trajectory of the density as a whole,
% pose the question of how close it remains to the ideal one.

% In order to characterize this notion of robustness, we lift to the
% space of probability distributions the notion of ISS, first introduced
% by~\cite{EDS:89}. This notion captures the effect of input
% disturbances on a dynamical system in a convenient and efficient way.
Introduced in~\cite{EDS:89} for deterministic dynamical systems in
finite-dimensional vector spaces, Input-to-State Stability (ISS) has
been extended to multiple domains.  For stochastic
systems, % several approaches have
% been proposed, for continuous SDEs
NSS~\cite{HD-MK-JW:01,LC-ZJ-EDS:25-arXiv}, and its $p$-th moment generalization
($p$thNSS)~\cite{DMN-JC:14}, alongside $p$-th moment ISS~\cite{LH-XM:09}
serve as a probabilistic counterpart of ISS for
continuous SDEs; while probabilistic ISS
(pISS)~\cite{PC-RKC-MT-ADM:23} applies to discrete-time systems. ISS
has also been extended to the infinite-dimensional
realm~\cite{SD-AM:12}, with applications to PDE
control~\cite{IK-MK:19} and large-scale swarm
control~\cite{TZ-QH-HL:21,TZ-QH-HL:21b}. In this case, the
distribution of agents is viewed as an element of an $L_p$ space. A
drawback of this approach is that it does not capture the geometry of
measures ($L^p$ just computes ``vertical distances'' while particles
move in the domain), and fails to assess the representation of finite
swarms by discrete measures.

In recent years, the theory of Optimal Transport (OT) has been a
fast-growing subject in applications that involve analysis in
probability spaces. Originally posed as an optimization problem, under
certain conditions~\cite{CV:08,FO:01} OT provides a geometric
characterization of a probability space as a Riemannian manifold.  A
class of functional optimization problems over this space leads to the
definition of Wasserstein gradient flows~\cite{LA-NG-GS:08}, which are
linked to classical continuity
equations.  % One can also understand the space as a
% Riemannian manifold~\cite{FO:01} under appropriate conditions.
This geometric interpretation of OT has been used for (mean-field)
control in different scenarios, such as stochastic linear system
steering~\cite{YC-TTG-MP:15}, Model Predictive
Control~\cite{ME-BB:23}, coverage control~\cite{VK-SM:22}, and safe
swarm control via Banach Barrier
functions~\cite{XG-GP-SB-SM:25}.

%\margin{we could add estimation to
%  mention the work by the prof from UWash, bayesian inference wasserstein}

In addition, the OT geometry allows for the generalization of notions
of convexity, gradient dominance, and smoothness to functionals over
probability spaces, which are fundamental in optimization.  Thus,
these tools are widely used in Machine Learning, for instance to
steer samples of a certain distribution to those of a
target~\cite{AW:18}.  In this field, the analysis of algorithm
perturbations is done by computing convergence
ratios~\cite{SDM-EN-SV:25} or asymptotic
bounds~\cite{CW-JDL-QL-TM:19}.  Instead, here we aim to provide a
control-theoretic, stability perspective as in~\cite{ES:22}, for
finite-dimensional gradient flows, by accounting for the Wasserstein
geometry. % This topic has also been explored in the context of

\textit{Statement of Contributions:} In this manuscript, we introduce
the notion of distributional Input-to-State Stability (dISS) for
continuous-time gradient Wasserstein flows, establish associated
Lyapunov characterizations, and prove its applicability for a class of
$l$-smooth functionals. We show how dISS unifies and strengthens the
notions of ISS for deterministic systems, and of NSS for stochastic
systems. Then, we demonstrate its applicability in the context of swarm
control, specifically for swarms subject to entropic disturbances, and
approximated dynamics based on agent-density approximations (via KDEs
and discrete-measures). This results into a characterization of the
mean-field objective error that is incurred when using a finite number
of agents.  Finally, we validate our results through simulations.

\paragraph*{Notation} Let $\Omega \subset \real^d$ be a
compact domain with boundary $\partial \Omega$ and unit outward normal
vector $n$, denote by $\bar{\Omega}$ its closure, and the identity
mapping on $\Omega$ as $\text{Id}$.  For any $x \in \Omega$ and a
subset $M \subset \bar{\Omega}$, the distance from $x$ to $M$ is given
by $\dist(x, M) :=\inf_{y \in M} |x - y|$, where $|\cdot|$ is the
  Euclidean norm in $\real^d$. We denote the set of
probability measures over $\Omega$ as $\mathcal{P}(\Omega)$, using
$\mathcal{P}_2(\Omega)$ for those with finite second-order moment. The
Lebesgue measure on $\real^d$ is denoted by $\vartheta$, and for any
$\rho \in \mathcal{P}(\Omega)$ we use $\rho$ to refer both to the
measure and its density function when $\rho \ll \vartheta$; that is,
when $\rho$ is absolutely continuous with respect to the Lebesgue
measure.  Given $T:\Omega \to \Omega$, $T_{\#}\rho=\nu$ denotes the
push forward of the measure $\rho$ towards $\nu$ under $T$; that is,
for any measurable subset $B$ of $\Omega$, $\nu(B)=\rho(T^{-1}(B))$.
The sets $\mathcal{C}_b(\Omega)$, $\mathcal{C}^\infty(\Omega)$ and
$\mathcal{C}_c(\Omega)$ represent the space of bounded continuous
functions, infinitely differentiable functions, and those with compact
support in $\Omega$, respectively. For $p \geq 1$, the space
$L^p(\Omega; \diff\rho)$, or simply $L^p(\rho)$, is the space of
measurable functions $f$ such that that
$\|f\|_{L^p(\rho)} := (\int_\Omega |f|^p \diff\rho)^{1/p}$. Finally,
$L^\infty_c$  denotes compactly supported
$L^\infty$-densities.

\section{Preliminaries}\label{sec:prel}
Here, we review the main concepts regarding the Wasserstein space of
probability distributions; see~\cite{FS:15,MB-GK:15} for more details.
Let $\Omega \subset \real^d$ be a domain, and $\Ptwo$ the set of
probability measures over $\Omega$ and finite $2$-nd
order moments.
%%% this is weak convergence %%%
% To define convergence in this space, we use the notion of weak
% convergence.
%\begin{definition}[Weak convergence of distributions]\marginguillem{should we erase the results on weak convergence?
%they are relevant when characterizing continuity, but we don't really mention it in the paper, then we can talk 
%directly of W2}
% A sequence of distributions $\{\rho_n\}_{n\geq1}$ such tat
% $\rho_i \in \Ptwo, \forall i\geq 1$, is said to converge weakly to
% $\rho^* \in \Ptwo$, (noted as $\rho_n \rightharpoonup \rho^*$) if
% $\forall \phi \in \mathcal{C}_b (\Omega)$:
  %\begin{align*}
  %  \lim_{n\to\infty} \int_\Omega \phi\, \diff\rho_n \to \int_\Omega \phi\, \diff\rho^*.
% \end{align*}\oprocend
%\end{definition}
The space $\Ptwo$ is endowed with the following metric.
\begin{definition}\longthmtitle{Wasserstein Metric}
  The $2$-Wasserstein distance between two measures
  $\rho, \nu \in \Ptwo$ is given by
\begin{equation*}
  W_2(\rho, \nu) = \inf_{\gamma \in \Pi(\rho,\nu)}
  \left( \int_{\Omega \times \Omega} |x - y|^2 \diff\gamma(x,y) \right)^{1/2},
\end{equation*}
where $\Pi(\rho,\nu)$ denotes the set of probability distributions
over $\Omega \times \Omega$ with marginals $\rho$ and $\nu$. \oprocend
\end{definition}

%The next result connects the Wasserstein metric and weak convergence.
%\begin{lemma}[{\cite[Theorem 5.11]{FS:15}}]
%Let $\Omega\subset
%^n$ be a compact set,
%then the following holds:
% \begin{equation*} 
%\rho_n \rightharpoonup \rho^*
%\iff W_2(\rho_n,\rho^*)\rightarrow 0.
%\end{equation*}\oprocend
%\end{lemma}

Endowed with it, $(\Ptwo,W_2)$ defines the $2$-Wasserstein metric
space. The distance value is directly related to the optimal transport
problem, and the original Monge formulation to solve it, which is
\begin{equation}\label{eq:monge}
  \text{(MP)}\quad  C(\rho,\nu) =
  \inf_{T: T_{\#}\rho = \nu} \left( \int_\Omega c(x ,T(x))  \diff\rho(x) \right). 
\end{equation}

%A more approachable formulation of the optimal transport problem is the Kantorovich relaxation:
%\begin{equation}\label{eq:kantorovich}
%  \text{(KP)}\quad \inf_{\gamma \in \Pi(\rho,\nu)} \left( \int_{\Omega \times \Omega} c(x,y)\diff\gamma(x,y) \right),
%\end{equation}
Results regarding existence and uniqueness of optimal transport maps
$T$ over non-compact domains can be found in~\cite{AF-AF:10}, which
require absolute continuity of $\rho$ and certain regularity
conditions on $c$. % which are satisfied by $c(x,y) = |x-y|^2$.
%A relevant result is the following.
%\begin{theorem} [{\cite[Theorem 1.17] {FS:15}}]
%  Given $\rho$ absolutely continuous and $\nu$ probability measures on
 % a compact domain $\Omega \subset \real^n$ with a zero measure border
 % $\partial \Omega$, there exists an optimal transport plan for the
  %cost $c(x, y) = h(x - y)$ with $h$ strictly convex. The map is
 % unique and there exists a function $\varphi$, named Kantorovich
 % potential, such that it has the form
 % $T(x) = x - (\nabla h)^{-1}(\nabla\varphi(x))$. The optimal coupling
  %then corresponds to $(\Id, T)_{\#} \rho$.\oprocend
%\end{theorem}
In particular, these hold
for %$W_2^2$ is the solution of the optimal transport problem~\eqref{eq:monge}
$c(x,y) = |x-y|^2$, which makes $C(\rho,\nu) =
W_2^2(\rho,\nu)$. Denote by $T_{\nu\to\rho}$ the optimal transport map
from an absolutely continuous $\nu$ to another $\rho$. Given two
$\rho_0,\rho_1 \in \Ptwo$ a generalized geodesic
$\mu_t:[0,1]\to \Ptwo$ with base measure $\nu$, is
$\mu_t=((1-t)T_{\nu\to\rho_1}+t T_{\nu\to\rho_0})_{\#}\nu$. When
$\nu$ is taken to be $\rho_0$, the generalized geodesic reduces to
a (standard) geodesic.
%Also note that existence of solutions
%to the Kantorovich formulation does not imply existence of solutions
%to the Monge formulation \eqref{eq:monge}, but the opposite does
%hold. In that case, $W_2^2$ is also the solution to the Monge
%formulation.
We define a tangent space in the $\Ptwo$ as follows. 
% note that it can be defined for any measure, not nec abs cont
\begin{definition}
  \longthmtitle{Tangent Space {\cite[Definition 8.4.1]{LA-NG-GS:08}}}
  Let $\rho \in \Ptwo$.  The tangent space at $\rho$,
  $T_\rho \Ptwo \subseteq L^2(\Omega;\diff \rho)$, is defined as the closure
  \begin{equation*}
    T_\rho \Ptwo 
    := \overline{\{ \nabla \phi : \phi \in C^\infty(\Omega)\cap C_c(\Omega) \}}^{L^2(\Omega;\diff \rho)} .
  \end{equation*}
  Thus, elements of $T_\rho \Ptwo$ are gradient vector
  fields $\nabla\phi$ which are $L^2(\Omega;\diff \rho)$
  integrable. \oprocend
\end{definition} 

\begin{definition}
  \longthmtitle{Inner Product} An inner product can be induced from
  that of $L^2(\Omega;\diff \rho)$, to $T_\rho \Ptwo$ as:
    \begin{equation*}
      \langle v_1, v_2 \rangle_\rho
      \triangleq \int_\Omega \langle v_1, v_2 \rangle \diff\rho, \quad v_1, v_2 \in T_\rho \Ptwo.
    \end{equation*}
    This inner product induces the $L^2(\Omega;\diff \rho)$ norm,
    $\|\cdot \|_{\rho}$.\oprocend
  %\begin{equation*}
  %      \|v\|_{T_\rho \Ptwo} := \|v\|_{L^2(\Omega;\diff \rho)} = \left( \int_\Omega |v(x)|^2 \diff\rho(x) \right)^{1/2}.
  %  \end{equation*} \oprocend
\end{definition}

A functional in the Wasserstein space is a mapping of the form
$F : \Ptwo \longrightarrow \real$. Important
properties of functionals include the following, and can be found in
{\cite[Chapter 7]{FS:15}}.
%\begin{definition}[ Regular Functional]
%  $F$ is regular if for all $\rho \in \Ptwo \cap L^\infty_c$ and
%  $\epsilon \in [0,1]$: 
%  \begin{equation*}
%    F((1-\epsilon)\rho + \epsilon \bar{\rho}) < \infty, \quad
%    \forall \bar{\rho} \in \Ptwo \cap L^\infty_c.
%  \end{equation*}
%  \oprocend
%\end{definition}

\begin{definition}\longthmtitle{First Variation}\label{def:first-variation}
  Given a functional $F$, the first variation of $F$ at $\rho_0$,
  denoted $\frac{\delta F}{\delta \rho}|_{\rho_0}$, is defined as any
  measurable function that satisfies:
  \begin{equation*}  
    \frac{d}{d\epsilon} F(\rho_0 + \epsilon \chi) \bigg|_{\epsilon=0}
    = \int_{\Omega} \frac{\delta F}{\delta \rho}|_{\rho_0}\, d\chi,
  \end{equation*}
  for every perturbation $\chi = \tilde{\rho} - \rho_0$ with
  $\tilde{\rho} \in \Ptwo \cap L_c^\infty$. \oprocend
\end{definition}

%\begin{definition}[Directional Derivative]
%  The derivative of a functional $F$ at $\rho_0$ along a tangent
%  vector field $v : \Omega \to \real^n$ is:
%  \begin{align*}
%    D_v F(\rho_0) =
%    \langle\nabla \frac{\delta F}{\delta \rho}|_{\rho_0},v\rangle_{\rho_0}.
%  \end{align*} \oprocend
%\end{definition}

%\begin{definition}[Geodesic Convexity in $\Ptwo$]
%  A functional $F(\rho)$ is said to be geodesically convex if, for any
%  geodesic $\rho_t:=((1-t)Id+tT_{\rho_0\rightarrow\rho_1})_{\#\rho_0}$
%  between $\rho_0,\rho_1$, $F(\rho_t)$ is convex in $t$. \oprocend
%\end{definition}

\begin{definition}\longthmtitle{$\lambda$ Geodesic
    Convexity}\label{def:lambda-convexity}\label{def:gen-convex}
  A functional $F$ is said to be $\lambda$-\emph{generalized 
  geodesically convex}  with respect to a base $\nu$ if for any %absolutely continuous
  measure $\rho_0 \in \Ptwo$ and $\rho_1\in\Ptwo$ it holds that:
 %\begin{equation}
 %  F(\rho_t)\leq(1-t)F(\rho_0)+tF(\rho_1)-\frac{\lambda}{2}t(1-t)W_2^2(\rho_0,\rho_1).
 %\end{equation}
     %      Equivalently:
  
{\small 
  \begin{equation}\label{eq:lambda-convex}
    F(\rho_1) \geq F(\rho_0) 
    +\langle \nabla \frac{\delta F}{\delta \rho}|_{\rho_0}(T_{\nu\to\rho_0}),v\rangle_{\nu}+ 
    \frac{\lambda}{2} W_{2,\nu}^2(\rho_0,\rho_1),
  \end{equation}
}

\noindent where $v = T_{\nu\to\rho_1} - T_{\nu\to\rho_0}$ is the
initial velocity of the generalized geodesic from $\rho_0$ to $\rho_1$
with base measure $\nu$, and $W_{2,\nu}^2(\rho_0,\rho_1)$ corresponds to the transport
cost through the base measure:
\begin{equation*}
   W_{2,\nu}^2(\rho_0,\rho_1)=\int_\Omega |T_{\nu\to\rho_1}(x) - T_{\nu\to\rho_0}(x)|^2\diff\nu.
\end{equation*}
Note that $W_{2,\nu}(\rho_0,\rho_1)\geq W_2(\rho_0,\rho_1)$ always.
The second term defines the directional
derivative $D_v F(\rho_0)$, which employs the inner product of $T_{\rho_0} \Ptwo$.
\oprocend
\end{definition}

A functional is \textit{$\lambda$-displacement convex} if the
expression in~\eqref{eq:lambda-convex} holds for regular geodesics.
That is, for any pair of densities $\rho_1,\rho_0$, we take
as base measure one of them. We refer as $\lambda$-convex to functionals
that are either generalized geodesically or displacement convex.
  % couplings instead of transport maps, one avoids the restriction to
  % singular measures, \margin{do you mean the restriction to abs cont
  %   measure?} we keep the presented definition for clarity.} 
%\margin{I comment this out; it's too much info}

%The previous definition extends to \textit{$\lambda$ generalized
%geodesic convexity} via the concept of generalized geodesic
%w.r.t.~a base (absolutely continuous) measure $\nu$.  Let $\rho_0$,
%$\rho_1 \in \Ptwo$, and let $T_{\nu\to\rho_0}$ (resp.~$T_{\nu\to\rho_1}$) be the optimal
%transport from $\nu$ to $\rho_0$ (resp.~$\rho_1$). Then the
%inequality~\eqref{eq:lambda-convex} should hold for $v = T_{\nu\to\rho_1}  - T_{\nu\to\rho_0}$.
% \change{For
% generalized geodesics w.r.t.~a base measure $\nu$,
% $\rho_t:=((1-t)T_0+tT_1)_{\#\rho_0}$ where $T_{0\#\nu}=\rho_0$, and
% $T_{1\#\nu}=\rho_1$, the functional is defined to be\lambda$
% generalized geodesically convex.} 
In the following, we assume that
$F$ is differentiable (in the Wassertein sense) and lower
semicontinuous, which implies that the first variation is well defined
and its gradient belongs to the tangent space. We assume that the set
of minimizers $\mathfrak{P}^*$ of
$F$ is non-empty, and we denote $F^*=F(\mathfrak{P}^*)$.
\begin{definition}\longthmtitle{Quadratic Growth}\label{def:quad-growth}
  The functional $F$ satisfies a quadratic growth inequality if the
  following holds:
    \begin{equation*}
        W_2^2(\rho,\rho^*)\leq\frac{2}{\lambda}(F(\rho)-F(\rho^*))\tag*{\oprocend}
    \end{equation*}
\end{definition}
%\begin{remark}
%    The quadratic growth for the entropy functional is the \emph{Talagraand Inequality ($T_2$)}
%\end{remark}

\begin{definition}\longthmtitle{Gradient Dominance}\label{def:grad-dom}
  The functional
  $F$ satisfies gradient dominance if the following holds:
  \begin{equation*}
      \Big\|\nabla\frac{\delta F}{\delta \rho}|_{\rho_0}\Big\|^2_{\rho_0}\geq2\lambda(F(\rho_0)-F(\rho^*)).\tag*{\oprocend}
  \end{equation*}
\end{definition}

%\begin{remark}
%    The gradient dominance for the entropy functional is the \emph{Logarithmic Sobolev Inequality (LSI).} \oprocend
%\end{remark}
\begin{lemma}\label{le:convex}
  If the functional $F$ is either $\lambda$ geodesically convex
    or $\lambda$-displacement convex, it
  satisfies both quadratic growth and gradient dominance.\oprocend
\end{lemma}
A sketch of the proof can be found in Appendix~\ref{sec:app-le}.
A notion of Lipschitz gradient smoothness in Wasserstein
space can be defined as follows.

\begin{definition}\longthmtitle{$l$-smoothness} \label{def:lsmooth} A functional
  $F$ is said to be $l$-smooth if for any %absolutely continuous
  $\rho_0\in\Ptwo$ and $\rho_1 \in \Ptwo$: 
  \begin{equation*}
   |F(\rho_1)-F(\rho_0)-D_v F(\rho_0)|\leq \frac{l}{2} W_{2,\nu}^2(\rho_0, \rho_1),
  \end{equation*}
  where $v = T_{\nu\to\rho_1} - T_{\nu\to\rho_0}$, for some base $\nu$.
  The $l$-smoothness is \textit{local} when the constant $l$
  depends on the bounded region of the space.\oprocend

\end{definition}

%\begin{lemma}\label{le:l-smoothness}
%  Let $F$ be a $l$-smooth functional. Then the following holds:
%  \begin{equation*}
%    \int_\Omega \langle\nabla\frac{\delta F}{\delta \rho}|_{\rho_1}-\nabla\frac{\delta F}{\delta \rho}|_{\rho_0},v\rangle\diff\nu \leq l W_{2,\nu}^2(\rho_0,\rho_1),
%  \end{equation*}
%  whre $v=T_{\nu\to\rho_1} - T_{\nu\to\rho_0}$.
%\end{lemma}
\begin{definition}\longthmtitle{Lipschitz Functional}\label{def:loc-lipschitz}
 A functional $F$ is globally Lipschitz if for any pair of distributions
 $\rho_0,\rho_1$ the following holds:
 \begin{equation*}
 	|F(\rho_0)-F(\rho_1)| \leq L W_2(\rho_0,\rho_1).
 \end{equation*}
 If the constant $L$ depends on the bounded region of space, the functional is locally Lispchitz.
\end{definition}

Finally, we recall the following definitions.

\begin{definition}[Comparison functions]
  The set of functions $\mathcal{K}$ consist of the set of increasing
  functions $\alpha : \realnonnegative \to \realnonnegative$ that are
  continuous, and satisfy $\alpha(0) = 0$. In particular,
  $\mathcal{K}_\infty \subset \mathcal{K}$ are functions that are
  radially unbounded. % , i.e., $\alpha(r) \to \infty$ as
  % $r \to \infty$.
  %The sets $\mathcal{K}$ and $\mathcal{K}_\infty$ are closed under sums, products, and compositions. Functions in $\mathcal{K}_\infty$ are invertible, and satisfy $\alpha^{-1} \in \mathcal{K}_\infty$.
  The functions in $\mathcal{KL}$ have two arguments
  $\beta: \realnonnegative\times \realnonnegative \to
  \realnonnegative$, and are s.t.~for fixed $t$,
  $\beta(r,t)\in\mathcal{K}$, while for fixed $r$, $\beta(r,t)$ is a
  decreasing function such that $\beta(r,t)\to 0$, as
  $t\to\infty$.\oprocend
\end{definition}

\begin{definition}[Size function]
  Let $(X, d)$ be a metric space and $M \subset X$ be a closed set.  A
  continuous function $w : X \to \real_{\ge 0}$ is a \emph{size
    function} with respect to $M$ if there exist
  $\underline{\alpha}, \overline{\alpha} \in \mathcal{K}_\infty$ such
  that:
  \begin{equation*}
    \underline{\alpha}(d(x, M)) \le w(x) \le \overline{\alpha}(d(x, M)), \quad \forall x \in X. \tag*{\oprocend}
  \end{equation*}
\end{definition}
%\begin{example}
% Any power $w_p(x) = (d(x, M))^p$ with $p > 0$ is a size function, as one can choose $\underline{\alpha}(s) = \overline{\alpha}(s) = s^p \in \mathcal{K}_\infty$.
%\end{example}
\section{Problem formulation}
We aim to identify conditions under which limited disturbances
affecting the transport of very large swarms guarantee a certain
stable behavior w.r.t.~the unperturbed ideal trajectories. To do this,
we rely on and adapt the notion of Input to State Stability (ISS) for
Wasserstein space, and apply this notion to continuity equations.
% Wasserstein (gradient) flows, as
% well as discrete-time iterations in the $W_2$ space.  We start by
% describing

\textit{\textbf{Continuity Equations in Wasserstein space.}}  Consider
an absolutely continuous trajectory of densities
$\rho_t: (a,b) \to \Ptwo$. Let $v_t\in L^2(\Omega;\diff \rho_t)$, and
such that $\|v_t\|^2_{L^2(\rho)} \leq|\rho_t'|$, for each $t$, where
$|\rho_t'|$ is the metric derivative of $\rho_t$ in Wasserstein
space.\footnote{The metric derivative of $\rho_t$ is the limit
  $|\rho'_t|=\lim_{s \rightarrow t}\frac{W_2(\rho_s,\rho_t)}{|s -t|}$,
  which exists a.e.~$t$ as the curve of measures is absolutely
  continuous.}  Then, the trajectory $\rho_t$ satisfies the following
\textit{continuity equation (PDE) in Wasserstein
  space}~\cite[Theorem~8.3.1]{LA-NG-GS:08}:\footnote{In other words,
  for an appropriate set of smooth functions with compact support in
  $\Omega$, $\varphi$, it holds that
  $\int_{(a,b)} \int_\Omega (\partial_t \, \varphi + \langle v_t,
  \nabla \varphi \rangle )\diff \rho_t \diff t = 0$.}
\begin{equation}\label{eq:cont-flow}
\begin{cases}
  \partial_t\rho_t=-\nabla\cdot (\rho_tv_t) , & \text{in } \Omega, \\[6pt]
  \rho_t (v_t \cdot n) = 0, & \text{on } \partial \Omega.
\end{cases} 
\end{equation}
In particular, the PDE flow is a \textit{gradient flow of a
 functional $G$ on $\Ptwo$}~\cite{AF:24}, when $v_t =-\nabla \phi$, for
$\phi = \frac{\delta G}{\delta \rho}$, the first variation of $G$;
cf.~Definition~\ref{def:first-variation}.  \oprocend

% An
% absolutely continuous trajectory $\rho_t : (a,b) \to \Ptwo, $ is the
% solution (in the sense of distributions\footnote{This means that, for
%   an appropriate set of smooth functions with compact support,
%   $\varphi$, it holds that
%   $\int_{(a,b)} \int_\Omega (\partial_t \, \varphi + \langle v_t,
%   \nabla \varphi \rangle )\diff \rho_t \diff t = 0$.})  to the
% following \textit{(Wasserstein) transport equation}
% PDE~\cite[Theorem~8.3.1]{LA-NG-GS:08}:
% \begin{equation}\label{eq:cont-flow}
% \begin{cases}
%   \partial_t\rho_t=-\nabla\cdot (v_t \rho_t) , & \text{in } \Omega, \\[6pt]
%   \rho_t (v_t \cdot n) = 0, & \text{on } \partial \Omega,
% \end{cases} 
% \end{equation}
% where $v_t :\Omega \rightarrow \real^d$,
% $v_t\in L^2(\Omega;\diff \rho_t)$, and such that
% $\int_\Omega \|v_t\|^2 \diff \rho_t < |\rho_t'|$, for each $t$, where
% $|\rho_t'|$ is the metric derivative of $\rho_t$.  \footnote{The
%   metric derivative of $\rho_t$ is given as the limit
%   $\lim_{s \rightarrow t}\frac{W_2(\rho_s,\rho_t)}{|s -t|}$, which
%   exists a.e.~$t$ from absolutely continuity of the curve of
%   measures.} 
In the following, we consider velocity fields $v_t$ that are affected
by disturbance signals, $u_t \in \mathcal{U}$. Given
$v_t \in L^2(\Omega, \diff \rho_t)$, we will assume that the corresponding
solutions to~\eqref{eq:cont-flow} exist, and that the boundary
condition holds for all time. The trajectory is then given by the
expression $\rho_t(\rho_0,u)$, which we denote as $\rho_t$ when clear
from the context.  As~\eqref{eq:cont-flow} is a conservation law, if
the support of $\rho_0$ is in $\Omega$ and $v_t \cdot n = 0$, then
$\rho_t$ is a well-defined distribution over $\Omega$.

In multi-agent systems,~\eqref{eq:cont-flow} are the mean-field limit
of microscopic dynamics, as explained in~\cite{RWB:12}. In particular,
given $f:\Omega \times \mathcal{U} \rightarrow \real^d$, it holds
that:
\begin{equation} \label{eq:flow-deterministic}
    \begin{gathered}
        \dot{x}=f(x,u) \\
        x(0)=x_0\sim\rho_0
    \end{gathered}
    \quad \longleftrightarrow\quad
    \begin{gathered}
        \partial_t\rho_t=-\nabla\cdot(\rho_tf(x,u)) \\
        \rho_0(x) = \rho(0,x).
    \end{gathered}
  \end{equation}
  Where $x_0\sim\rho_0$ means $x_0$ is sampled from the distribution
  $\rho_0$.  The right direction follows by taking the large-agent
  limit on the equations on the left, while the left one
  follows via sampling. For stochastic dynamics, the evolution of both
  the microscopic and macroscopic states evolve as:
  \begin{equation}\label{eq:flow-stochastic}
    \resizebox{\hsize}{!}{$
      \begin{array}{@{}l@{}}
        dx \!=\! f(x)dt \!+\! g(x) \Sigma(t)  dw \\
        x(0) = x_0\!\sim\! \rho_0
      \end{array}
      \,\leftrightarrow\,
      \begin{array}{@{}l@{}}
        \partial_t\rho \!=\! -\nabla \!\cdot\! \big( \rho \big( f \!-\! \frac{1}{2\rho}\nabla\cdot(Q \rho) \big) \big) \\
        \rho_0(x) \!=\! \rho(0,x).
      \end{array}
      $}
  \end{equation}
  On the left, we have a stochastic differential equation where $dw$
  is a standard Wiener process, and $Q=g\Sigma\Sigma^\top
  g^\top$. The result, grounded in robotic swarms, is presented in~\cite{HH-HW:08} 
  and rewritten here as a continuity equation.

\begin{problem}\label{problem}
  Consider a perturbed version of~\eqref{eq:flow-deterministic}, where
  we replace $v_t$ by a smooth vector field,
  $\Xi: \Ptwo \times \mathcal{U}\to\real^d$, where
  $u \in \mathcal{U}$ are disturbances, and $\Xi(\rho,0)$ is the
  unperturbed velocity field. Do the following:
  \begin{enumerate}
  \item Provide a disturbance attenuation notion in the ISS sense
    (distributional ISS) that is applicable for Wasserstein flows.
    (Section~\ref{sec:diss-notions}).
  \item Evaluate the relationship of dISS with existing disturbance
    attenuation notions in the ISS sense for deterministic and
    stochastic microscopic systems. (Section~\ref{sec:iss-relation}).
    
    % \margin{Question: one thing is the
      % notion itself, another one is the characterization via an ISS
      % lyapunov function. Are the notions for discrete and continuous
      % ``the same'' but the characterizations in terms of
      % iterations/cont flows different? }\marginguillem{Actually, the
      % characterization is also pretty much equivalent, as we can
      % convert an ISS lyapunov function to a dISS Lyapunov function by
      % considering its expected value with respect to the state
      % density. the positivity bounds on the lyapunov functions are
      % equivalent to V being a proper loss function, this is explained
      % in Remark 5.4. }
  \item Study the dISS properties of perturbed Wasserstein gradient
    flows. (Section~\ref{sec:grad-flows}). \oprocend
  \end{enumerate} 
\end{problem}

% \textit{Problem 1} Determine the conditions under which the
% perturbations of the trajectories $\rho_t$ obtained
% from~\eqref{eq:cont-flow} remain close to the original flows in the
% sense of Input-to-State-Stability (ISS) with respect to a
% state-dependent input signal.
% \oprocend\\
% Our goal is to derive a notion of distributional ISS (dISS), in
% particular for continuous trajectories, which will allow for the
% analysis of perturbed flows in the Wasserstein space. To derive
% specific conditions to guarantee dISS, we focus on the particular
% scenario of perturbed Wasserstein gradient flows, which model
% several cases such as the stochastic noise on individual agents.
% \\
% In order to compare the effect of disturbances on the individual
% trajectories and on the collective distribution we are also
% interested in the characterization of the relation between the new
% notion of dISS and already existing notions, in particular ISS for
% the case of deterministic trajectories, and NSS for the case of
% stochastic trajectories.

% Furthermore, we also consider the counterpart of the previous problem
% for discrete-time perturbed trajectories.

\section{dISS for Continuity Equations}
%In this section, we aim to provide a solution for Problem 1.
%Specifically we present the notion of distributional ISS (dISS) in
%Section~\ref{sec:diss-notions}, we evaluate its relation with particle
%notions of ISS in Section~\ref{sec:iss-relation}, and we study
%perturbed gradient flows in Section~\ref{sec:grad-flows}.
We start by addressing the first two items of Problem~\ref{problem}.
\subsection{dISS notions}\label{sec:diss-notions}
We consider the perturbed version of a continuous flow in
\eqref{eq:cont-flow}, which yields:
\begin{equation}\label{eq:pert-cont-flow}
  \partial_t\rho_t=-\nabla\cdot (\rho_t \,\Xi(\rho_t,u_t)),
\end{equation}
where $\Xi: \mathcal{U}\times\Ptwo\to\real^d$ is a smooth vector
field. Here, $\Xi(\rho,0)$ corresponds to the unperturbed velocity
field that results in~\eqref{eq:cont-flow}.
Before introducing our ISS notion, we discuss
the suitability of an existing $L^2(\Omega)$-ISS.

\begin{remark}[Relationship with $L^2(\Omega)$-ISS]
  An approach to establish the ISS properties of very-large swarms
  leverages the $L^2(\Omega)$ norm~\cite{TZ-QH-HL:21,TZ-QH-HL:21b}.
  While this benefits from the Banach structure of $L^2(\Omega)$, this
  concept excludes distributions that are not absolutely continuous
  w.r.t.~the Lebesgue measure, such as Dirac deltas, which are the
  distributional representation of finite agent groups. In addition,
  the small sensitivity of $\|\cdot \|_{L^2}$ to small-measure sets
  and mass shifts results into $L^2$ distances that are not able to
  distinguish well between distributions that may differ significantly
  from a target measure.  For example, Figure~\ref{fig:l2-w2} shows
  $\rho_1,\rho_2$ with a non-overlapping support with that of a target
  density $\rho^*$. Even though their supports are separate, they are
  equally close from $\rho^*$ in the $L^2$ sense.
  %Similarly, two step-like signals
 % $\rho_1, \rho_2$ with overlapping support have the same
  %$sL^2(\Omega)$ distance to a third $\rho^*$. \margin{draw this, two
  %  steps with differnet locations for the step as compared to a
  %  constant signal} As the shapes differ significantly, sampling from
  %each of these distributions can result into a very different, even
  %inconsistent, finite agent deployments. 
   In contrast, the Wasserstein distance can explicitly account for the
  spatial displacement of mass, which motivates the introduction of
  dISS.\oprocend
\end{remark}
\begin{figure}[htbp]
    \centering
    \begin{tikzpicture}
        \begin{axis}[
            width=\linewidth,   
            height=5.5cm,         
            axis x line=bottom,
            axis y line=none,
            ymin=0, ymax=2.2,
            xmin=-1, xmax=17,   
            ticks=none,
            xlabel={},
            ylabel={}, 
            axis line style={-stealth},
            label style={font=\small}, 
        ]

        \addplot [
            domain=-1:5, 
            samples=150, 
            color=red!80!black, 
            fill=red!20, 
            very thick           
        ] 
        { (x > 0 && x < 5) ? 
          (1.8 * exp(-0.5*((x-2)/0.5)^2) * (1 / (1 + exp(-4*(x-2))))) : 0 };

        \node[red!80!black, font=\small] at (axis cs: 2.3, 1.3) {$\rho_{t_1}$};

        \addplot [
            domain=4:11, 
            samples=100, 
            color=red!80!black, 
            fill=red!20, 
            very thick
        ] 
        { (abs(x-7.5) < 3) ? 0.8 * exp(-0.5*((x-7.5)/0.7)^2) : 0 };

        \node[red!80!black, font=\small] at (axis cs: 7.5, 0.95) {$\rho_{t_2}$};

        \addplot [
            domain=10:17, 
            samples=100, 
            color=blue!80!black, 
            fill=blue!20, 
            very thick
        ] 
        { (abs(x-13.5) < 3.5) ? 0.5 * exp(-0.5*((x-13.5)/1.0)^2) : 0 };
        
        \node[blue!80!black, font=\small] at (axis cs: 13.5, 0.65) {$\rho^*$};

        \end{axis}
    \end{tikzpicture}
    
    \caption{\small The swarm relaxes from an initial skewed,
      high-concentration state $\rho_{t_1}$ (left), through a
      transitional state $\rho_{t_2}$ (middle), moving towards the
      target distribution $\rho^*$ (right). In this case
      $W_2(\rho_{t_2},\rho^*)< W_2(\rho_{t_1},\rho^*)$; however
      $\|\rho_{t_1}-\rho^*\|_{L^2} = \|\rho_{t_2}-\rho^*\|_{L^2}$, due
      to non-overlapping supports.}
    \label{fig:l2-w2}
\end{figure}
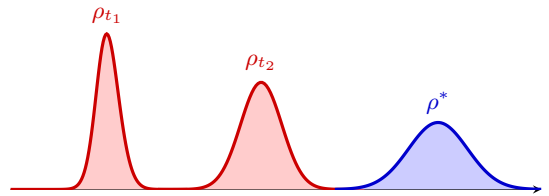

% As a consequence, the
%bounds for the dynamic density over time guaranteed by the ISS
%properties become more conservative.  \oprocend

%Motivated by this, we introduce a distributional Input to State
%Stability (dISS) for continuous-time systems.
\begin{definition}[Distributional ISS]\label{def:diss}
  Let $\mathfrak{P}^* \in \mathcal{P}_2(\Omega)$ be the set 
  of stationary points of the unperturbed flow
  $\partial_t \rho_t = -\nabla \cdot (\rho_t \Xi(0,
  \rho_t))$. Consider a trajectory $\hat{\rho}_t$ generated by the
  perturbed continuity equation \eqref{eq:pert-cont-flow} under the 
  input signal $u$. Then, the perturbed flow is said to be
  distributional ISS (dISS) if there exist
  $\beta\in\mathcal{KL}$, $\gamma\in\mathcal{K}$, such that:
\begin{equation*}
    W_2(\hat\rho_t,\mathfrak{P}^*)\leq \beta(W_2(\hat\rho_0,\mathfrak{P}^*),t)+\gamma(\|u\|_c),
\end{equation*}
$\forall \hat{\rho}_0\in\Ptwo$. Here,
$\|u\|_c=\sup_{\tau\geq0}(\|u_\tau\|_{\mathcal{U}})$, and we 
use $W_2(\rho, \mathfrak{P})\coloneqq \inf_{\nu \in \mathfrak{P}} W_2(\rho, \nu)$. \oprocend
%We adopt this notation in the
%rest of the paper.
\end{definition}

Standard ISS is characterized via ISS Lyapunov functions. We
introduce here its counterpart for dISS.
\begin{definition}\longthmtitle{dISS Lyapunov
    Functional}\label{def:dLISS}
  Consider the perturbed flow~\eqref{eq:pert-cont-flow}. Assume
  $ \rho^*\in\Ptwo $, and let
  $ V : \Ptwo \rightarrow \realnonnegative $ be a differentiable
  functional. We say that $ V $ is a \emph{dISS Lyapunov functional}
  if there are $ \psi_1, \psi_2 \in \mathcal{K}_\infty $,
  $ \chi \in \mathcal{K} $, and a continuous positive-definite
    function $ \alpha $, for which the following holds for all $t\geq 0$:
  \begin{enumerate}
  \item Positivity and distance equivalence.
    \begin{align} \label{eq:lyap-bounds} \psi_1(W_2(\rho_t,
      \mathfrak{P}^*))& \leq V(\rho_t) \leq \psi_2(W_2(\rho_t,
      \mathfrak{P}^*)).
    \end{align}
  \item Decay condition. If
  %  \begin{equation*}
   $   W_2(\rho_t, \mathfrak{P}^*) \geq \chi(\|u_t\|_\mathcal{U}),$
  %  \end{equation*}
    then %the time derivative of $ V $ satisfies
    \begin{equation} \label{eq:lyap-decay}
      \dot{V}_u(\rho_t) \leq -\alpha(W_2(\rho_t, \mathfrak{P}^*)),
    \end{equation}
    where $V_u(\rho_t) \equiv V(\rho(t;\rho_0,u))$, with a slight
    abuse of notation. The time derivative is computed as:
  %\begin{align*}
  %  \dot{V}_u(\rho) = \frac{d}{dt} V(\rho_t) \bigg|_{\partial_t \rho = -\nabla \cdot(\rho_t\Xi(u_t,\rho_t))}\\
  %  =\int_{\Omega} \nabla \left( \frac{\delta V}{\delta \rho} \right) \cdot \Xi(u_t,\rho_t)  d\rho.
  %\end{align*}
    \begin{equation*}
      \dot V_u(\rho_t)=\frac{d}{dt} V(\rho_t) \bigg|_{-\nabla \cdot(\rho_t\Xi(u_t,\rho_t))}=D_{\Xi(u_t,\rho_t)}V(\rho_t).
    \end{equation*}
  \end{enumerate}
\oprocend
\end{definition}

With this, dISS admits the following characterization. 

\begin{theorem}\label{thm:diss-lyap}\longthmtitle{Distributional ISS via Lyapunov Functional}
  Assume that for all input signal $u$, $u_t \in \mathcal{U}$ and all $s \geq 0$, the
  shifted input $\tilde{u}$ defined by $\tilde{u}_t \coloneq u(t+s)$
  satisfies $\tilde{u}_t \in \mathcal{U}$ and
  $\|\tilde{u}\|_c \leq \|u\|_c$.  Then, if there exists a dISS
  Lyapunov functional $ V : \Ptwo \rightarrow \real^+ $ satisfying the
  conditions of Definition~\ref{def:dLISS} the perturbed Wasserstein
  gradient flow is distributional ISS (dISS).
\end{theorem}

The proof translates the ISS Lyapunov function characterization for
infinite dimensional settings as presented in \cite{SD-AM:12}. The key
technical point relies on the fact that the solutions
of~\eqref{eq:cont-flow} are absolutely continuous, and, thus, they are
continuous w.r.t.~the metric $W_2$. For the sake of completeness, we
sketch its main steps in Appendix~\ref{sec:app-proof-lyap}.

\begin{corollary}
  A sufficient condition for a $V(\rho)$ to be a dISS Lyapunov
  functional is given by  the positivity bounds in~\eqref{eq:lyap-bounds}
  alongside the following inequality:
  \begin{equation*}
    \dot{V_u}(\rho_t)\leq -\chi(W_2(\rho_t, \mathfrak{P}^*))+ \gamma(\|u_t\|)\quad\forall t\geq0,
  \end{equation*}
  where $\chi,\gamma$ are both class $\mathcal{K}_{\infty}$ functions.
\end{corollary}
\begin{remark}[Local dISS]
  %While the definition above formulates dISS globally (holding for all
  %$\rho_0 \in \mathcal{P}_2(\Omega)$ and all inputs), many physical
  %systems satisfy this property only locally. 
  A system is \emph{locally dISS} if there exists a forward invariant
  set $\mathcal{S}\subset \Ptwo$ and $K_u > 0$ such that dISS holds only when
  the initial error and input satisfy $\rho_0\in \mathcal{S}$,
  $\|u\|_c \leq K_u$. This restricts the domain of validity
  of~\eqref{eq:lyap-decay} to the region defined by $\mathcal{S}$.\oprocend
\end{remark}
 
\subsection{Relationship of dISS with microscopic ISS and NSS}
\label{sec:iss-relation}
Here, we establish the relationship of dISS with single-agent notions
of ISS for deterministic systems~\cite{EDS:89} and NSS for stochastic
systems~\cite{LC-ZJ-EDS:25-arXiv}. In $\real^d$, (N-)ISS characterizes
the stability towards a certain set $M\subset \real^d$. 
% while
% dISS adds a layer on top for swarms, and characterizes the stability
% towards a certain density with support on a set. 
Thus, in order to bridge this gap, % we need to
% flatten that dimension, that is,
we will limit dISS to target distributions that are supported on that
set,
\begin{equation*}
  \mathfrak{P}*=\PM \equiv \PtwoM\coloneqq\setdef{\rho}{ \text{supp}(\rho) \subset M}.
\end{equation*}
This target set reduces dISS to an ISS notion on the expected value of
the state. To see this, we use the  Monge
formulation~\eqref{eq:monge} of the distance to the set $\PM$:
\begin{equation*}
    W_2^2(\rho, \PM) = \inf_{T(x)\in M} \int_{\Omega} |x-T(x)|^2\diff\rho.
\end{equation*}
For this problem, the optimal transport map exists and
is given by $T(x) = \arg\min_{y\in M} \|x-y\|^2$, which yields,
\begin{equation}\label{eq:exp-w2}
\small
W_2^2(\rho, \PM)=\int_{\Omega}\dist^2(x,M)\diff\rho
=\mathbb{E}_{x\sim\rho}[\dist^2(x,M)].
\end{equation}
When the state is a Dirac delta, $\rho=\delta_x$, the Wasserstein
distance to the set $\PM$ reduces to
$W_2(\delta_x, \PM) = \dist(x, M)$, allowing us to draw a direct
connection between dISS and classical ISS for deterministic systems.

%\margin{This is a todo, but I believe that contractivity can also be
% defined using W2 and then it will also translate into the
%  contractivity notions of deterministic and stochastic systems in the
%  same way as we do here. It is something we can either put in a
%  remark or consider separately in an extended paper. Not a top
 % priority, though}

\begin{proposition}
  Given a set $\Omega \subset \real^d$ and a compact
  $M \subset \Omega$, the microscopic
  dynamics~\eqref{eq:flow-deterministic} are ISS with respect to~$M$
  if the density flow in~\eqref{eq:flow-deterministic} is dISS with
  respect to $\PM$. Furthermore, the converse also holds whenever the
  $\beta\in\mathcal{KL}$ function is linear in the domain variable.
\end{proposition}

\begin{proof}
  \textbf{($\text{dISS} \implies \text{ISS}$):} Take
  as initial measure a Dirac delta centered
  at the initial particle condition, $\rho_0 = \delta_{x_0}$. The
  evolution of this distribution is $\rho_t = \delta_{x_t}$, where
  $x_t$ is the particle flow. The dISS condition guarantees:
  \begin{equation*}
    W_2(\delta_{x_t}, \PM) \leq \beta(W_2(\delta_{x_0}, \PM), t) + \gamma(\|u\|_c).
  \end{equation*}
  Using the identity $W_2(\delta_x, \PM) = \dist(x, M)$
   we recover the particle ISS formulation:
  \begin{equation*}
    \dist(x_t, M) \leq \beta(\dist(x_0, M), t) + \gamma(\|u\|_c).
  \end{equation*}
  
  \textbf{($\text{ISS} \implies \text{dISS}$):} The converse relies on
  the definition of the Wasserstein distance to the set $\PM$. Using
  the change of variables formula (push-forward) for the flow map of
  the particle dynamic system, $\Phi_t(x)$, we have:
  
  {\small
  \begin{equation*}
  W_2^2(\rho_t,\PM)=\int_\Omega \dist^2(x,M)\diff\rho_t
  =\int_\Omega \dist^2(\Phi_t(x),M)\diff\rho_0 .
  \end{equation*}
  }
  Since $\dist^2(\cdot, M)$ is a size function, 
  the ISS property implies the existence $\beta\in\mathcal{KL}$
  and $\gamma\in\mathcal{K}$ such that:
  \begin{equation*}
    \dist^2(\Phi_t(x), M) \leq \beta(\dist^2(x, M), t) + \gamma(\|u\|_c).
  \end{equation*}
  Plugging this bound directly into the integral:
  \begin{equation*}
    \begin{aligned}
      \int_\Omega \dist^2(\Phi_t(x), M) \diff\rho_0 \leq \int_\Omega
      \Big( &\beta(\dist^2(x, M), t) \\
            &+ \gamma(\|u\|_c) \Big) \diff\rho_0.
    \end{aligned}
  \end{equation*}
  We now utilize the linearity of $\beta \in \mathcal{KL}$ to
  interchange the integral sign.
  \begin{equation}\label{eq:jensen}
    \int_\Omega \beta(\dist^2(x, M), t) \diff\rho_0
    \leq L(t) \int_\Omega \dist^2(x, M) \diff\rho_0 .
  \end{equation}
  Connecting back with the Wasserstein distance identity
  $W_2^2(\rho_0, \PM)= \int_\Omega \dist^2(x, M) \diff\rho_0$, we
  obtain the desired dISS bound:
  \begin{equation*}
    W_2(\rho_t, \PM) \leq \beta_d(W_2(\rho_0, \PM), t) + \gamma_d(\|u\|_c).
  \end{equation*}
where we defined
    $\beta_d (r,t)=r \sqrt{L(t)} $, and
    $\gamma_d=\sqrt{\gamma}$.
\end{proof}

\begin{remark}[Compact domains]\label{re:compact}
  When $\Omega$ is compact and $\beta(\cdot,t)$ has a finite
  right-derivative at the origin, we can bound
  $\beta(r,t)\leq L(t) r$, where
  $L(t)=\sup_{x\in\Omega} \frac{\beta(|x|,t)}{|x|}$.  Alternatively,
  without imposing conditions at the origin, one can always
  upper-bound the $\mathcal{KL}$ function on a compact domain by
  another $\mathcal{KL}$ function $\sigma(r,t)$ that is concave in its
  first argument, and then apply Jensen's inequality
  in~\eqref{eq:jensen}.  Thus, in compact domains one obtains the
  equivalence between both ISS notions.  Note that the latter
  implication refers to a \textit{set of distributions of any shape}
  over $M$, which can be far wider than desired, see Figure~\ref{fig:domain}.\oprocend
\end{remark}
  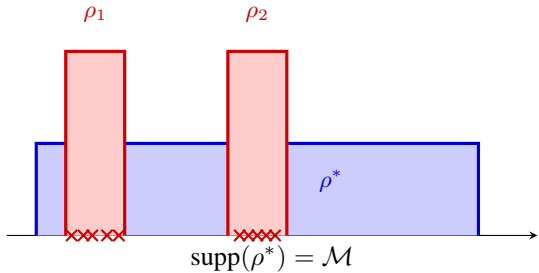
\begin{figure}[htbp]
      \centering
      \begin{tikzpicture}
          \begin{axis}[
              width=\linewidth,   
              height=5.5cm,         
              axis x line=bottom,
              axis y line=none,
              ymin=0, ymax=0.8,
              xmin=-1, xmax=17,   
              ticks=none, 
              xlabel={$\text{supp}(\rho^*)=\mathcal{M}$},
              ylabel={}, 
              axis line style={-stealth},
              clip=false
          ]

          % 1. TARGET DISTRIBUTION (CONSTANT SIGNAL)
          % Width = 16, Height = 0.0625 (Area = 1)
          %\addplot [
          %    domain=0:16, 
          %    samples=2, 
          %    color=blue!80!black, 
          %    opacity=0.5,
          %    fill=blue!20, 
          %    very thick
          %] 
          %{0.0625}; 

          \addplot[
              blue!80!black, 
              fill=blue!20, 
              very thick, 
              const plot,
          ] coordinates {(0,0) (0,0.25) (15,0.25) (15,0)};
          
          % Target Label
          \node[blue!80!black, font=\small] at (axis cs: 10, 0.15) {$\rho^*$};

          % 2. SOURCE STEP FUNCTIONS
          % Left Step
          \addplot[
              red!80!black, 
              fill=red!20, 
              very thick, 
              const plot,
          ] coordinates {(1,0) (1,0.5) (3,0.5) (3,0)};
          
          % Left Step Label
          \node[red!80!black, font=\small] at (axis cs: 2, 0.6) {$\rho_{1}$};
          
          % Right Step
          \addplot[
              red!80!black, 
              fill=red!20, 
              very thick, 
              const plot,
          ] coordinates {(6.5,0) (6.5,0.5) (8.5,0.5) (8.5,0)};
          
          % Right Step Label
          \node[red!80!black, font=\small] at (axis cs: 7.5, 0.6) {$\rho_{2}$};
          
          % 3. FINITE AGENT DEPLOYMENTS (SAMPLES ON HORIZONTAL AXIS)
          % Left samples (Red crosses on y=0)
          \addplot[
              only marks, 
              mark=x, 
              mark size=3pt, 
              color=red!80!black,
              thick
          ] coordinates {(1.2, 0) (1.6, 0) (1.9, 0) (2.4, 0) (2.8, 0)};
          
          % Right samples (Red crosses on y=0)
          \addplot[
              only marks, 
              mark=x, 
              mark size=3pt, 
              color=red!80!black,
              thick
          ] coordinates {(6.9, 0) (7.2, 0) (7.5, 0) (7.8, 0) (8.1, 0)};

          \end{axis}
      \end{tikzpicture}
      
      \caption{\small Two step functions ($\rho_1$ and $\rho_2$) compared against a constant target distribution
       ($\rho^*$). The crosses on the horizontal axis represent finite agent samples drawn
      from each step distribution, illustrating the severe spatial displacement that occurs despite
      both being supported in $\mathcal{M}$, and having constant $L^2$ error.}
      \label{fig:domain}
\end{figure}

For stochastic dynamics~\eqref{eq:flow-stochastic}, the particle
notion of ISS extends to Noise to State Stability (NSS), where the
unknown covariance is treated as the input disturbance. In particular,
NSS holds if there exists a size function of the set $M$, $w_M$, and
$\beta\in\mathcal{KL}$, $\gamma\in\mathcal{K}$, such that for all
$\epsilon\in(0,1)$:
\begin{equation*}
  \mathcal{P}\{w_M(x_t)\leq \beta(w_M(x_0),t)+\gamma(\|\Sigma\Sigma^\top\|_\infty)\}\geq 1-\epsilon .
\end{equation*}
% \begin{equation}\label{eq:flow-stochastic}
%     \resizebox{\hsize}{!}{$
%     \begin{array}{@{}l@{}}
%         dx \!=\! f(x)dt \!+\! \Sigma(t) g(x) dw \\
%         x(0) \!\sim\! \rho_0
%     \end{array}
%     \,\Leftrightarrow\,
%     \begin{array}{@{}l@{}}
%         \partial_t\rho \!=\! -\nabla \!\cdot\! \big( \rho \big( f \!+\! \frac{1}{2\rho}\nabla\cdot(Q \rho) \big) \big) \\
%         \rho(x,0) \!=\! \rho_0(x)
%     \end{array}
%     $}
% \end{equation}
% where $dw$ is a standard Wiener process, and
% $Q=g\Sigma\Sigma^Tg^T$. In this case
In distributional space, the perturbed functional is more complex,
with
$\Xi(\rho,u)=-\frac{1}{2\rho}\nabla
(Q\rho)$; see~\eqref{eq:flow-stochastic}.
\begin{proposition}\label{prop:nss-diss}
  Given a set $\Omega \subset \mathbb{R}^n$ and a compact target
  $M \subset \Omega$, the particle flow in \eqref{eq:flow-stochastic}
  is NSS with respect to the set $M$ if the density flow is dISS with
  respect to the set $\PM$. The converse does hold when $\Omega$ is
  compact or there exists a $\supscr{p}{th}$-NSS Lyapunov function. \oprocend
\end{proposition}
The proof for compact $\Omega$ uses a dISS Lyapunov approach and is
provided in Appendix~\ref{sec:app-proof-nss}.  For unbounded domains
one can rely on the connection between $\supscr{p}{th}$-NSS Lyapunov
functions and 2nd-moment bounds; see~\cite{DMN-JC:14}.

  \section{Perturbed Gradient Flows}\label{sec:grad-flows}
%   % In this section, we study the robustness of gradient flows in
%   %  Wasserstein space. % In the context of 
%  mean-field control systems, these trajectories are of great use to
%  stabilize a fully-actuated swarm towards a target distribution in
%  an optimal way.

  Consider a differentiable functional $G$ on $\Ptwo$, and its first
  variation, $\phi=\frac{\delta G}{\delta \rho}$. Assuming that
  $\nabla \phi \cdot n_{|\partial \Omega} = 0$, where $n$ the outward
  normal to $\partial \Omega$, recall that the \textit{gradient flow
    of $G$ in Wasserstein space} can be expressed as:
  \begin{equation}\label{eq:cont-grad-flow}
    \partial_t \rho_t=-\nabla\cdot(\rho_t (-\nabla \frac{\delta G}{\delta \rho}|_{\rho_t}))=\nabla\cdot(\rho_t\nabla\phi) .
  \end{equation}
  Let $\widehat{\rho}_t$ be a trajectory of the perturbed flow:
  \begin{equation}\label{eq:pert-cont-grad-flow}
    \partial_t\widehat{\rho}_t=\nabla\cdot
    (\widehat{\rho}_t (\nabla\phi+\zeta_{u_t}(\widehat{\rho}_t))),
  \end{equation}
  where $u_t \in \mathcal{U}$ is a time-varying bounded input and 
  $\zeta_u(\rho)$ is a perturbation vector field
  depending on both the state and the input.
  We provide conditions on $G$ and the
  state-dependent perturbation in~\eqref{eq:pert-cont-grad-flow} so
  that the gradient flow is dISS.
\begin{proposition}\label{prop:iss-cont}
  Suppose that the functional $G$ has a compact set of minimizers
  $\mathfrak{P}^*$, and that it satisfies quadratic growth, gradient
  dominance, and local $l$-smoothness along regular geodesics (Section~\ref{sec:prel}).  Let
  $\zeta_u(\rho)$ be bounded,
  $\|\zeta_u(\rho)\|^2_{L^2(\Omega;\diff \rho)}\leq \gamma(\|u\|)$,
  where $\gamma\in \mathcal{K}_\infty$. 
  Then, the perturbed flow in~\eqref{eq:pert-cont-grad-flow} is dISS.
\end{proposition}
\begin{proof}
  Let $G$, the functional in~\eqref{eq:cont-grad-flow}, be the dISS
  Lyapunov functional candidate. The positivity bounds are satisfied
  due to quadratic growth and local $l$-smoothness, which together
  guarantee:
  \begin{equation*}
    \frac{\lambda}{2}W_2^2(\rho,\mathfrak{P}^*)
    \leq G(\rho)-G^* \leq \alpha_l(W_2^2(\rho,\mathfrak{P}^*)),
  \end{equation*}
  where $\alpha_l$ is a $\mathcal{K}_\infty$ function that bounds the
  growth of the local $l$-smoothness.  Its derivative along the
  perturbed flow becomes:

  {\small
  \begin{equation*}
    \dot{G_u}=-\int_\Omega
    \nabla\phi\cdot(\nabla\phi+\zeta_u(\rho))\diff\rho=
    -\int_\Omega(|\nabla\phi|^2+\nabla\phi\cdot \zeta_u(\rho))\diff\rho,
  \end{equation*}
  }
  Applying Young's inequality to the second term:
  \begin{equation}\label{eq:young-cont}
    \dot{G_u}\leq-\frac{1}{2}\int_\Omega |\nabla\phi|^2d\rho
    +\frac{1}{2}\int|\zeta_u(\rho)|^2\diff\rho.
  \end{equation}
  Now, from gradient dominance and quadratic growth inequalities to
  the first term yields:
\begin{equation*}
  -\int_\Omega |\nabla\phi|^2\diff \rho\leq
  -2\lambda(G(\rho)-G^*)\leq-\lambda^2 W_2^2(\rho,\mathfrak{P}^*).
\end{equation*}
Substituting in expression~\eqref{eq:young-cont}, and using the
boundedness of $\zeta_u(\rho)$, we obtain:
\begin{align*}
\dot{G}_u(\rho) &\le -\frac{1}{2}\lambda^2 W_2^2(\rho,\mathfrak{P}^*)
+ \frac{1}{2}\gamma(\|u\|).
\end{align*}
\end{proof}
\begin{remark}\longthmtitle{Relationship with (N-)ISS
    finite-dimensional gradient flows and
    applicability} %[Relationship with ISS]
% For this particular
% functional, a know result is that $G(\rho)$ is geodesically convex iff
% $\phi(x)$ is convex. \margin{references?} Thus, the dISS
% behavior of the trajectory can be studied through the following
% proposition.
  The conditions imposed on $G$ are analogous to those for
  finite-dimensional perturbed gradient flows~\cite{ES:22}. In
  particular, the previous result is directly applicable to
  $G(\rho)=\mathbb{E_\rho}[V(x)]$ when $V(x)$ is Lipschitz continuous
  and strongly-convex; as then $G$ is strongly geodesically
  convex~\cite[Chapter 7]{FS:15} (implying quadratic growth and
  gradient dominance).  However, we only need $V(x)$ to be a
  \textit{proper loss function} for the result to hold, as the next
  %The result provides a distributional counterpart
  %for stochastic gradients to be NSS~\cite{LC-ZJ-EDS:25-arXiv}.
  \oprocend
\end{remark}

\begin{lemma}\label{le:proper-loss}
  Let $V: \Omega \to \mathbb{R}$ be a proper loss function with
  respect to a compact set $M$ and assume convexity
  of the lower bound comparison functions.
  Define the functional
  $\mathcal{V}(\mu) = \int_{\Omega} V(x) \diff \mu$ for
  $\mu \in \Ptwo$. Then $\mathcal{V}$ satisfies quadratic growth,
  gradient dominance, and local $l$-smoothness
  (Section~\ref{sec:prel}), which means Proposition~\ref{prop:iss-cont}
  is applicable.\oprocend
\end{lemma}

We provide the proof in Appendix~\ref{sec:app-proof-proper}, which
includes the definition of proper loss function given in~\cite{ES:22}.

\begin{remark}[Comparison to the $L^2$ Geometry]
  The stability analysis crucially relies on the geometry of the
  Wasserstein space. If the density $\rho$ is treated as an
  element of $L^2(\Omega)$, the functional
  $\mathcal{V}(\mu) = \int_\Omega V(x)\diff \mu$ becomes
  linear with respect to standard measure interpolation,
  $\mathcal{V}((1-t)\rho_0 + t\rho_1) = (1-t)\mathcal{V}(\rho_0) +
  t\mathcal{V}(\rho_1)$, see Figure~\ref{fig:interpolation}. In other words, the functional has zero
  curvature in the $L^2$ geometry.  This linearity prevents the
  satisfaction of two critical properties in
  Proposition~\ref{prop:iss-cont}, which are quadratic growth and
  gradient dominance.  At the same time, restricting the set of probability
  distributions with bounded upper and lower values:
  \begin{equation*}
    \mathcal{X}=\{\rho\in\Ptwo\;|\; 0<m\leq\rho(x) \leq M<\infty, \; x \in \Omega \text{ a.e.}\},
  \end{equation*}
  the $W_2$ distance is upper bounded by the $L_2$-norm~\cite{RP:18},
  $W_2(\rho,\nu) \le C\|\rho-\nu\|_{L^2}$.  In this case, proving
  $L^2$-ISS is more conservative than dISS.\oprocend
  \label{re:l2-proof}
\end{remark}

\begin{remark}[Relaxation of $l$-smoothness condition]
  The $l$-smoothness of the functional serves two purposes.
  First, it guarantees the upper bound of the functional in terms
  of the $W_2$ distance to the minimizer (in the discrete case
  it will be fundamental to prove the decay condition too). Second, it guarantees that 
  the solutions to the continuity equation are unique in the weak sense.
  While the former can still be guaranteed relaxing the condition to 
  local Lipschitzness (Def. \ref{def:loc-lipschitz}), the latter is harder to ensure.
Furthermore, relevant functionals in the Wasserstein Space, such as KL divergences,
  do not satisfy $l$-smoothness nor local Lipschitsness.
  However, there is a way to guarantee that  those functionals are also robust in the dISS sense 
  if they are $\lambda$-convex, as we prove in the following.
\end{remark}
\begin{proposition}
   Suppose that the functional $G$ is $\lambda$-convex
   and has a unique minimizer $\rho^*$.  Let
  $\zeta_u(\rho)$ be bounded,
  $\|\zeta_u(\rho)\|^2_{L^2(\Omega;\diff \rho)}\leq \gamma(\|u\|)$,
  where $\gamma\in \mathcal{K}_\infty$. 
  Then, the perturbed flow in~\eqref{eq:pert-cont-grad-flow} is dISS.
\end{proposition}
\begin{proof}
  We take as a Lyapunov functional $\frac{1}{2}W_2^2(\rho,\rho^*)$. Its derivative
  along the flow evaluates to:
  \begin{equation}\label{eq:convex-cont}
  \begin{aligned}
    &\frac{d}{dt}\frac{1}{2}W_2^2(\rho_u,\rho^*)
     =\int_\Omega\langle-(\nabla\phi+\zeta_u(\rho)),x-T_{\rho\to\rho^*}(x)\rangle\diff\rho \\
    &=\int_\Omega\langle\nabla\phi,T_{\rho\to\rho^*}(x)-x \rangle\diff\rho
    +\int_\Omega\langle\zeta_u(\rho),T_{\rho\to\rho^*}(x)-x \rangle\diff\rho.
  \end{aligned}
  \end{equation}
  For the first term we use the $\lambda$-convexity of the functional, if the functional is $\lambda$-displacement convex, we have:
  \begin{equation*}
    \int_\Omega\langle\nabla\phi,T_{\rho\to\rho^*}(x)-x \rangle\diff\rho\leq G(\rho^*)-G(\rho)-\frac{\lambda}{2}W_2^2(\rho,\rho^*).
  \end{equation*}
  As $G(\rho^*)$ is the minimizer we can upper bound the difference between the first two terms.
  For a generalized geodesically convex functional the previous expression can be rewritten by choosing
  as base measure $\rho^*$, and performing a change of variables:
  {\small
  \begin{equation*}
    \int_\Omega \langle\nabla\phi(T_{\rho^*\to\rho}),x-T_{\rho^*\to\rho}(x)\rangle\diff\rho^*\leq-\frac{\lambda}{2}W_{2,\rho^*}^2(\rho,\rho^*),
  \end{equation*}
  }
  which can be bounded by the standard distance.
  For the second term in~\eqref{eq:convex-cont} we use Cauchy and Young's inequality for $\lambda>\epsilon>0$:
  {\small
  \begin{equation*}
    \int_\Omega\langle\zeta_u(\rho),T_{\rho\to\rho^*}(x)-x \rangle\diff\rho\leq \frac{1}{2\epsilon}\int_\Omega|\zeta_u(\rho)|^2\diff\rho+\frac{\epsilon}{2}W_2^2(\rho,\rho^*).
  \end{equation*}
  }
  Putting both together we obtain the dISS bound:
  \begin{equation*}
    \frac{d}{dt}\frac{1}{2}W_2^2(\rho_u,\rho^*)\leq-\frac{(\lambda-\epsilon)}{2}W_2^2(\rho,\rho^*)+\frac{1}{2\epsilon}\gamma(\|u\|).
  \end{equation*}
  
\end{proof}

\begin{figure}[htbp]
    \centering
    \begin{tikzpicture}
        \begin{axis}[
            width=\linewidth,   
            height=4.5cm,
            axis x line=bottom,
            axis y line=none,
            ymin=0, ymax=1.0, 
            xmin=-6, xmax=6,   
            ticks=none, 
            xlabel={},
            ylabel={}, 
            axis line style={-stealth},
            clip=false,
            domain=-6:6,
            samples=150
        ]

        % 1. L^2 INTERPOLATION (Red)
        \addplot[
            color=red!80!black, 
            fill=red!20, 
            fill opacity=0.6,
            very thick
        ] {0.35 * exp(-0.5*(x+2)^2) + 0.35 * exp(-0.5*(x-2)^2)};
        
        \node[red!80!black, font=\small] at (axis cs: 0, 0.25) {$\rho_{L^2}$};

        % 2. WASSERSTEIN W_2 INTERPOLATION (Green)
        \addplot[
            color=green!50!black, 
            fill=green!20, 
            fill opacity=0.5,
            very thick
        ] {0.7 * exp(-0.5*(x)^2)};
        
        \node[green!50!black, font=\small] at (axis cs: 0, 0.8) {$\rho_{W_2}$};

        % 3. BASE DISTRIBUTION 1 (Blue Left)
        \addplot[
            color=blue!80!black, 
            fill=blue!20, 
            fill opacity=0.4,
            very thick
        ] {0.7 * exp(-0.5*(x+2)^2)};
        
        \node[blue!80!black, font=\small] at (axis cs: -2, 0.8) {$\rho_{0}$};

        % 4. BASE DISTRIBUTION 2 (Blue Right)
        \addplot[
            color=blue!80!black, 
            fill=blue!20, 
            fill opacity=0.4,
            very thick
        ] {0.7 * exp(-0.5*(x-2)^2)};
        
        \node[blue!80!black, font=\small] at (axis cs: 2, 0.8) {$\rho_{1}$};

        \end{axis}
    \end{tikzpicture}
    
    \caption{\small Comparison of Euclidean ($L^2$) and Wasserstein ($W_2$)
      interpolations of two overlapping Gaussian distributions
      ($\rho_0$ and $\rho_1$). The $L^2$ interpolation ($\rho_{L^2}$)
      linearly joints the densities. In contrast, the $W_2$
      interpolation ($\rho_{W_2})$ horizontally transports the mass,
      preserving the Gaussian structure.}
    \label{fig:interpolation}
\end{figure}
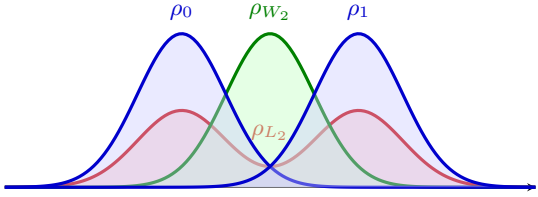

%\begin{remark}[l-smoothness]
%  The $l$-smoothness assumption in Proposition~\ref{prop:iss-cont},
%  may be too restrictive for general functionals. To
%  bridge this gap without losing the dISS guarantee, one can
%  alternatively consider the squared Wasserstein distance as the dISS
%  Lyapunov functional candidate.  This candidate trivially satisfies the
%  positivity bounds. Then, one only need to study under which
%  conditions its derivative under the gradient flow of the 
%  particular functional fulfills the required decay.  \oprocend
%\end{remark}

\subsection{Entropic Disturbances in Optimal Transport}
We present here applications of Proposition~\ref{prop:iss-cont}
stemming from swarm density control.

\begin{lemma}\longthmtitle{Continuity
    equation subject to additive entropic disturbance}
  Consider the advection-diffusion PDE:
\begin{equation}\label{eq:entr-dist}
    \partial_t\rho=\nabla\cdot(\rho\nabla\phi)+u\Delta\rho,
\end{equation}
where $u_t\in\real$ is a bounded input signal and $\phi$ is the first
variation of a functional $G$ satisfying the assumptions of
Proposition~\ref{prop:iss-cont}. Assume that there is a bound
for the norm of the logarithmic gradient of the density. That is:
  \begin{equation*}
    \|\zeta(\rho)\|_{L^2(\Omega;\diff \rho)}^2=\int_\Omega |\nabla \log(\rho)|^2 \diff\rho\leq M.
  \end{equation*}
  Then, the entropically perturbed gradient flow in
  \eqref{eq:entr-dist} is dISS.
\end{lemma}
%Here $\phi$ corresponds to the first variation of a functional $G$.
% Note how this equation corresponds to the mean field approximation of
% the stochastic particle system with isotropic noise of unknown
% covariance $\sqrt{2u}$:
% \begin{equation*}
%     dx=\nabla \phi(x) dt+\sqrt{2u} \Id dw,
% \end{equation*}

% begin{lemma}
%   Assume that the functional $G$ satisfies the conditions of Proposition~\ref{prop:iss-cont}, 
%   and that for all $t \geq 0$, the norm of its logarithmic gradient is bounded, that is:
%   \begin{equation*}
%     \|\zeta(\rho)\|_{L^2(\Omega;\diff \rho)}^2=\int_\Omega |\nabla \log(\rho_t)|^2 \diff\rho_t \leq M
%   \end{equation*}
%   then the entropically perturbed gradient flow in \eqref{eq:entr-pert} is dISS.
% \end{lemma}
\begin{proof}
  Observe that we can rewrite~\eqref{eq:entr-dist} as  the following
  perturbed gradient flow:
\begin{equation*}
  \partial_t\rho=\nabla·(\rho(\nabla\phi+u \zeta(\rho))),
  \quad \zeta(\rho):=\nabla \log(\rho).
\end{equation*}
% Note how the unknown covariance (the disturbance) in the stochastic particle system
% has been lifted to a disturbance at the distribution level.
In this way, the state dependency of the disturbance corresponds to
$\zeta(\rho)=\nabla\frac{\delta\mathcal{H}(\rho)}{\delta \rho}$, where
$\mathcal{H}(\rho)=\int_\Omega \log\rho \diff\rho$ is the entropy functional.
The result follows directly from Proposition~\ref{prop:iss-cont} by
taking the dISS Lyapunov functional candidate as $G(\rho)$, and using
the boundedness of the logarithmic gradient to bound the perturbation
term.
\end{proof}

\begin{remark}[Fisher Information Bound]\label{re:fisher}
  The condition in the previous lemma requires a bound on the term
  $I(\rho)=\int_\Omega |\nabla \log(\rho)|^2 \diff\rho$, the
  \textit{Fisher information integral}.  If $\rho\ll\vartheta$, the
  Fisher information reduces to
  $I(\rho)=\int_\Omega \frac{|\nabla \rho|^2 }{\rho}\diff \vartheta$.
  The bound holds if $\rho$ is strictly positive ($\rho(x)\geq m>0$)
  and $\nabla\rho\in L_2(\Omega)$ is bounded. On the other hand, if
  the first variation of the functional $G$ satisfies $\nabla^2 \phi \succeq \lambda\Id$,
  $I(\rho_t)$ is guaranteed to decrease exponentially along
  trajectories of~\eqref{eq:entr-dist}~{\cite[Thm. 24.2]{CV:08}}. 
  % \margin{is this lower bound
  % also guaranteed by the fact that $\Omega$ is convex? and from the
  % lower bound of $\rho$? (as we say after the proof of lemma
  % 6.4)?}\marginguillem{When G is the Wasserstein distance to a target, yes. But 
  % here we are saying for any functional, so depending on the functional the 
  % conditions change}
  Then, if the initial condition has bounded Fisher
  Information, it remains bounded $\forall t\geq 0$.\oprocend
\end{remark}

Observe that~\eqref{eq:entr-dist} is the mean field approximation of a
stochastic particle system with isotropic noise of unknown covariance
$\sqrt{2u}$.
%\begin{equation*}
%    dx=\nabla \phi(x) dt+\sqrt{2u} \Id dw,
%\end{equation*}
Therefore, the previous result implies that that steering a swarm via
optimal transport is robust to stochastic noise of unknown covariance
in the agents. This motivates another example
given by the entropic regularization of the optimal transport functional itself. 
\begin{lemma}[Regularized Optimal Transport]
  Consider the perturbed optimal transport functional:
  \begin{align*}
    W_{2,u}^2 (\mu,\nu) & = \inf_{\pi \in \Pi(\mu,\nu)}
                          \left[ \int_{\Omega \times \Omega} c(x,y) \diff \pi(x,y)
                          + u H(\pi)\right],\\
    H(\pi) & = \log \left[\frac{\diff \pi}{\diff \mu \diff \nu} \right],
  \end{align*}
  where $u$ is a bounded input signal, and let $G_u(\mu) = W_{2,u}^2(\mu,\rho^*)$, where $\rho^*$ is a
  target distribution. Define the associated gradient flow
  \begin{equation}\label{eq:perturbed-ot}
    \partial_t\rho
    =\nabla\cdot \left(\rho\nabla \phi_u\right),
  \end{equation}
  where $\phi_u= \frac{\delta G_u}{\delta \rho}$, and $\phi_0$ is the
  unperturbed potential.
  Assume that $\rho,\rho^*\ll \vartheta$, that $\Omega$ is a convex compact set
  and that $\rho,\rho^*$ are both upper and lower bounded for all
  $x\in\Omega$  ($0<m\leq\rho(x),\rho^*(x)\le M$).
  %Furthermore defining $\varphi(x)=\frac{1}{2}\|x\|^2-\phi_0(x)$,
  %assume \change{ A2) $\varphi\in C^2(\Omega)$ and
  %$\varphi^*\in C^{\alpha+1}(\Omega)$ with $\alpha>1$,\\
  %A3) $\mu \Id \preceq \nabla^2\varphi\preceq L\Id$ for
  %$\mu,L>0$.\\}
  Then, the gradient flow~\eqref{eq:perturbed-ot}
  %generated by the regularized optimal transport functional
  is dISS with respect to the unperturbed optimal transport flow.
  
  % , and let
  % $I_*(\rho,\rho^*)$ be the integral of the logarithmic gradient along
  % the Wasserstein geodesic $\{\rho_\tau\}_{\tau \in [0,1]}$, from
  % $\rho_0=\rho$ to $\rho_1=\rho^*$:
  % \begin{equation*}
  %   I_*(\rho,\rho^*)=\int_0^1\int_\Omega
  %   |\nabla \log(\rho_\tau)|^2 \diff\rho_\tau \diff \tau.
  % \end{equation*}
  % Then, if the transport maps are sufficiently smooth,
  % \margin{is this ensured with the strong convexity of the domain as
  % you point out below? then I would state it here} \marginguillem{not only
  % with that, thats why I dont give the conditions here}
  % and $I_*(\rho,\rho^*)$ is bounded along trajectories of
  % \eqref{eq:perturbed-ot}, \margin{that means that for every
  % $\rho_s$ trajectory of the perturbed trajectory, we have
  % $I(\rho_s,\rho^*)$ is bounded, right?}
  % \marginguillem{we need more than that, not only along the trajectory,
  % but along the geodesics departing from each point in the trajejctory}
\end{lemma}

% with respect to the unperturbed flow ($u=0$). 
% where $I_0(\rho,\rho^*)$ is the integral of the logarithmic gradient along
% the geodesic between $\rho$ and $\rho^*$:
% \begin{equation*}
%   I_0(\rho,\rho^*)=\int_0^1\int_\Omega |\nabla \log(\rho_t)|^2 \diff\rho_t \diff t,
% \end{equation*}

\begin{proof}
  Although the entropic regularization is inside the Wassertein
  functional itself, we can equivalently write the perturbed flow as
  \begin{equation*}
    \partial_t\rho=\nabla\cdot(\rho(\nabla\phi+\zeta_u(\rho))),
  \end{equation*}
  where $\zeta_u(\rho)=\nabla \phi _u-\nabla \phi$.  Assuming that
  transport maps are sufficiently smooth---which holds if $\rho$ and
  $\rho^*$ are sufficiently smooth---the following bound is given
  in~{\cite[Corollary 1]{AAP-JN:21},} (with more details on the
    constant $C$):
  \begin{equation*}
    \|\nabla \phi(\rho)-\nabla \phi_u(\rho)\|_{L^2(\Omega;\diff \rho)}^2
    \le u^{2}C.
  \end{equation*}
  Thus, the result follows directly from
  Proposition~\ref{prop:iss-cont} by taking the dISS Lyapunov
  functional $W_2^2(\rho,\rho^*)$.
\end{proof}

%The previous regularization is widely used with Sinkhorn algorithms
%due to its computational tractability. 
%Regarding the feasibility of
%these assumptions, (A2) and (A3) are theoretically guaranteed provided
%the domain $\Omega$ is strictly convex and the densities $\rho$ and
%$\rho^*$ are strictly bounded away from zero and sufficiently smooth,
%more details in~\cite{AAP-JN:21}.
%Hölder continuous ($\rho, \rho^* \in C^\beta(\Omega)$ for $\beta > 0$),
%guarantees (A3), (A2) necessitates stronger smoothness

%\subsection{Sample Approximations}
The implementation of mean field algorithms can not be done exactly as
the number of agents is finite. Instead, agents employ density
approximations, which can also be understood as a type of perturbation
to the original flow. Here, we establish dISS properties of the flows
associated with Kernel Density Estimators (KDEs) and particle
approximations.

\paragraph{KDEs flow approximation}
The KDE approximation of a density for $N$ agents is defined as
$\rho^{h,N}(x; z)=\frac{1}{N}\sum_{i=1}^N K_h(x-x_i)$, where $h$
is the bandwidth of the kernels, and $z \in \Omega^N$ represents the
agent locations, $z = (x_i)$, with $x_i$ the position of agent~$i$.
%If $x_i \sim \rho^l$, it holds that
%\begin{equation*}
%  \lim_{h\to 0, N\to \infty} \rho^{h,N}=\rho^l.
%\end{equation*}
We assume that $\rho^{h,N}$ is the actual state of the system.  Let
$G$ be a functional on $\Ptwo$, and denote %its evaluation at
% $\rho=\rho^{h,N}$ as
$G^{h,N}(z)\coloneqq G(\rho^{h,N}(x;z))$.

To minimize $G(\rho)$, a reasonable agent control input is:
\begin{equation}\label{eq:input-kernel} 
  \dot x_i= -\partial _{x_i} G^{h,N}(z) := -\frac{1}{N}(\nabla \phi \ast K_h)(x_i),
\end{equation}
where $\ast$ is a shorthand notation which reduces to standard
convolution for integral functionals, 
and $\phi$ the first variation of $G$ at $\rho^{h,N}$, see~\cite{VK-SM:22}
for details.
% In order to be well-defined, $G(\rho)$ needs to be differentiable
% (in the Wasserstein sense), and $\nabla G^{h,N} \cdot n = 0$ needs
% to hold to prevent mass from leaking.  Restricting the bandwith
% appropriately and assuming $G^{h,N}(z)$ is $\alpha$-smooth in the
% $z\in \Omega^N$ variable guarantees the convergence of the KDE to a
% local minimum~\cite{VK-SM:22}.
Under~\eqref{eq:input-kernel}, the KDE evolves as:
\begin{equation} \label{eq:kde-flow}
  \partial_t \rho = \nabla \cdot (\rho \nabla\phi^{h,N}),
\end{equation}
where $\nabla\phi^{h,N}$ can be obtained by expressing
$\partial_t\rho^{h,N}=\frac{1}{N}\sum_{i=1}^N\nabla_{x_i}K_h(x-x_i)(-\dot x_i)$
as a continuity equation, yielding:
\begin{equation}\label{eq:nadaraya}
  \nabla\phi^{h,N}(x)=
  \frac{\sum_{j=1}^{N}\partial _{x_j} G^{h,N}(z) K_h(x-x_j)}{\sum_{j=1}^N K_h (x-x_j)}.
\end{equation}
At points where $\rho^{h,N}(x) \to 0$, this expression is understood
in the limit sense. Notice how equation~\eqref{eq:nadaraya} takes the
form of a Kernel regression estimator of the convoluted ideal velocity
field~\cite{MPW-MCJ:94}.  We aim to characterize the dISS properties
of~\eqref{eq:kde-flow}, with respect to the ideal trajectory
in~\eqref{eq:cont-grad-flow}.
% \begin{equation*}

% 	\partial_t\rho=-\nabla \cdot (\rho \nabla \phi),
% \end{equation*}

% We interpret~\eqref{eq:nadaraya} as a Nadaraya-Watson estimator,
% \margin{provide a reference} and compute its Mean Square Error (MSE),
% $\mathbb{E}[|\nabla\phi^{h,N}-\nabla\phi|^2]$, which can be split into
% the variance and bias terms.

\begin{lemma}[KDE-based Continuity Equation is dISS]
  Consider the Wasserstein gradient flow of a functional
  $G$~\eqref{eq:cont-grad-flow}, and its
  counterpart~\eqref{eq:kde-flow} using a KDE approximation.
  Assume that the kernel $K_h$ satisfies the following properties:\\
  A1) 
  %The bandwith $h$ and the number of agents are coupled by 
  $h = c N^{-\frac{1}{d+2}}$, guaranteeing proper convergence of
  $\rho^{h,N}$ to $\rho$ in the mean-field limit, and
  minimizing the squared error of the estimator~\eqref{eq:nadaraya}. \\
  A2) It has finite first and second-order moments { \small
    $\int_\Omega K_h(z)\|z\|\diff z \triangleq h\mu_1(K_h)$,
    $\int_\Omega K_h(z)\|z\|^2 \diff z \triangleq h^2 \mu_2(K_h)$.}
  \\
  Suppose $G$ satisfies the conditions of
  Proposition~\ref{prop:iss-cont}, and assume that its first variation
  $\nabla \phi$ is Lipschitz continuous with constant $L$.
  Then~\eqref{eq:kde-flow} is dISS.
  %In particular, the velocity field error is upper bounded by:
  %\begin{equation*}
  %  \|\zeta_{h,N}\|^2_{L_2(\rho^{h,N})}\leq 2L^2(\mu_2(K_h)+\mu_1^2(K_h))  N^{-\frac{2}{d+2}}.
  %\end{equation*}
\end{lemma}

\begin{proof}
  To show dISS, we bound the $L_2(\rho^{h,N})$ norm of the
  perturbation $\zeta_{h,N} = \nabla\phi^{h,N}-\nabla\phi$.  Defining
  the convex combination coefficients
  $w_i(x)=\frac{K_h(x-x_i)}{N\rho^{h,N}(x)}$, we  write:
  \vspace{-1.5ex}
  \begin{equation*}
    \zeta_{h,N}=\sum_{i=1}^{N}w_i(x) e_i(x), %((\nabla\phi\ast K_h)(x_i)-\nabla\phi(x)).
  \end{equation*}
  where $e_i(x)=(\nabla\phi\ast K_h)(x_i)-\nabla\phi(x)$. Then,
  applying Jensen's inequality,
  $(\sum_i w_i e_i)^2 \le \sum_i w_i e_i^2$, and integrating
  with respect to $\rho^{h,N}$, yields:
  \begin{equation*}
    \|\zeta_{h,N}\|^2_{L_2(\rho^{h,N})}
    \leq \frac{1}{N}\sum_{i=1}^{N}\int_\Omega K_h(x-x_i) |e_i|^2 \diff x.
  \end{equation*}
  We bound the integrand by adding and subtracting $\nabla\phi(x_i)$
  to the term $e_i(x)$ and applying
  $\|a+b\|^2 \leq 2(\|a\|^2+\|b\|^2)$.  Because $\nabla \phi$ is
  $L$-Lipschitz, the first resulting term is bounded by the Kernel's
  second moment (A2):
  \begin{equation*}
    \int_\Omega K_h(x-x_i)|\nabla\phi(x_i)-\nabla\phi(x)|^2
    \diff x \leq L^2 h^2 \mu_2(K_h).
  \end{equation*} 
  The second term can be directly bounded by the Kernel's first moment
  (A2), giving:
  {\small
\begin{align*}
    |(\nabla\phi \ast K_h)(x_i) & - \nabla\phi(x_i)|^2 \int_\Omega K_h(x-x_i) \diff x \leq (L h \mu_1(K_h))^2.
          \end{align*}}
     
{\small
  \begin{align*}
    | (\nabla\phi \ast K_h)(x_i) & - \nabla\phi(x_i) |\\
                                 &  = | \int_\Omega (\nabla \phi(x) - \nabla \phi (x_i))K_h(x-x_i) \diff x |\\
                                 & \le L \int_\Omega |x-x_i| K_h(x-x_i) \diff x \\
    & = L h \mu_1 (K_h)
  \end{align*}
}
Combining these two bounds yields:
  \begin{equation*}
    \|\zeta_{h,N}\|^2_{L_2(\rho^{h,N})}\leq 2L^2(\mu_2(K_h)+\mu_1^2(K_h)) h^2.
  \end{equation*}
  The final bound in terms of $N$ follows immediately from the
  coupling $h = c N^{-\frac{1}{d+2}}$ in (A1).
\end{proof}

%\change{  The Variance is bounded as follows:
%    \begin{equation*}
%      \mathbb{E}[|\nabla\phi^{h,N} -\mathbb{E}[\nabla\phi^{h,N}]|^2]\leq C_1\frac{h^ 2}{Nh^d}+o(\frac{h^2}{Nh^d}),
%    \end{equation*}
%    while the bias can be split between the bias of the Nadaraya-Watson estimator~\cite{MPW-MCJ:94}
%    and the bias of the convolution:
%    \begin{equation*}
%      |\mathbb{E}[\nabla\phi^{h,N}]-\nabla\phi|^2\leq (C_2 h^2+C_3\frac{1}{Nh^d}+o(h^2)+o(\frac{1}{Nh^d}))^2.
%    \end{equation*}
%    Assuming that $h<1$ and $\frac{1}{Nh^d}<1$, defining
%    $u=(h^2,\frac{1}{Nh^d})$, and integrating over the domain
%    $\Omega$ we obtain the following bound:
%    \begin{align*}
%      \|\nabla\phi^{h,N}-\nabla\phi\|_{L^2(\Omega)}^2=\int_\Omega \mathbb{E}[|\nabla\phi^{h,N}-\nabla\phi|^2]\diff\rho\\
%      \leq u^TQu+M\|u\|_2^3=\gamma(\|u\|_2),
       %      \end{align*}}

  \paragraph{Semi-discrete Optimal Transport}
  An alternative to~\eqref{eq:input-kernel} relies on the approximation of the density via an empirical
  distribution $\rho^N=\frac{1}{N}\sum_{i=1}^N \delta_{x_i}$, which
  %results in the perturbed flow:
  %\begin{equation}\label{eq:delta-flow}
  %  \partial_t \rho^N=-\nabla \cdot (\rho^N v_N),
  %\end{equation}
  is steered by a velocity vector $v_N$ approximating the gradient flow of a certain
  functional.
  As in the previous case, there is some inherent error
  when we restrict our flow to the space of empirical measures due
  to the impossibility of the empirical distribution to match an absolutely continuous target.
  %Note that while this expression is written as a continuity
  %equation, in reality it reduces to a system of ODEs.
  We focus in this subsection on the particular functional given by
  $W_2(\rho^N,\rho^*)$ where $\rho^*\ll\vartheta$, known as
  Semi-Discrete Optimal Transport (SD-OT).  Note that this functional
  can be rewritten as follows:
  \begin{equation}\label{eq:sdot}
    W_2^2(\rho^N,\rho^*)=\sum_{i=1}^N \int_{W_i} |x-x_i|^2 \diff\rho^*,
  \end{equation}
  where $\{\mathcal{W}_i\}_{i=1}^N$ are the Voronoi-Laguerre cells,
  which define the following partition of $\Omega$:
  \begin{equation*}
    \mathcal{W}_i=\setdef{x\in \Omega}{ |x-x_i|^2-w_i\leq |x-x_j^^2|-w_j,\quad\forall j\neq i},
  \end{equation*}
  with $(w_i)$ being such that the mass of each cell is equal,
  $\rho^*(\mathcal{W}_i)=\frac{1}{N}\quad\forall i\in\{1,\dots,N\}$.
  The SD-OT particle flow is defined through the following input:
  \begin{equation}\label{eq:sdot-flow}
    \dot{x}_i=-\frac{N}{2}\nabla_{x_i}W_2^2(\rho^N,\rho^*)=-(x_i-N\int_{\mathcal{W}_i}x\diff\rho^*),
  \end{equation}
  we define the centroid of the Voronoi-Laguerre cell as
  $c_{\mathcal{W}_i}\coloneq N \int_{\mathcal{W}_i}x\diff\rho^*$.  Note
  here that even though the cells depend on the agent location, its
  effect on the derivative vanishes, as noted
  in~\cite{JC-SM-FB:05-esaim}. The stability of this particle flow is
  characterized via the following lemma.

  \begin{lemma}\label{le:dISS-semi-discrete}
    Consider an empirical distribution $\rho^N$ evolving under the 
    SD-OT flow in~\eqref{eq:sdot-flow}.
    Assume that the support of the target distribution
    is bounded. Then, the flow is dISS with respect to the target $\rho^*$,
    and the ultimate dISS
    bound is proportional to $N^{-2/d}$ for sufficiently large $N$.
  \end{lemma}

  \begin{proof}
    We define our Lyapunov function as the SD-OT functional.
    Taking the derivative along the continuous-time trajectories $\dot{x}_i$ yields:
    \begin{align}\label{eq:sdot-flow-decay}
      \frac{d}{dt} W_2^2(\rho^N, \rho^*) = \sum_{i=1}^{N}\nabla_{x_i}W_2^2(\rho^N,\rho^*) \dot{x}_i=\\
      \frac{2}{N} \sum_{i=1}^N \langle x_i - c_{\mathcal{W}_i}, \dot{x}_i \rangle \nonumber
      = -\frac{2}{N} \sum_{i=1}^N |x_i - c_{\mathcal{W}_i}|^2.
    \end{align}
    Expression~\eqref{eq:sdot-flow-decay} is related to the SD-OT functional in~\eqref{eq:sdot},
    which can be rewritten by adding and subtracting
    the centroid $c_{\mathcal{W}_i}$ inside the norm and squaring the sum:
    
    {\small
    \begin{align*}
      W_2^2(\rho^N, \rho^*)=\sum_{i=1}^N(|x_i - c_{\mathcal{W}_i}|^2 \int_{\mathcal{W}_i} \diff\rho^* \\
      + \int_{\mathcal{W}_i} |c_{\mathcal{W}_i} - x|^2 \diff\rho^*
      + 2 \left\langle x_i - c_{\mathcal{W}_i}, \int_{\mathcal{W}_i} (c_{\mathcal{W}_i} - x) \diff\rho^*\right\rangle).
    \end{align*}
    }
    
    Since $c_{\mathcal{W}_i}$ is the centroid of $\mathcal{W}_i$ the
    last term vanishes. Noting that the mass of each optimal transport cell is uniform ($\int_{\mathcal{W}_i} \diff\rho^* = 1/N$), 
    the SD-OT functional decomposes perfectly into:
    \begin{equation*}
      W_2^2(\rho^N, \rho^*) = \frac{1}{N} \sum_{i=1}^N |x_i - c_{\mathcal{W}_i}|^2 + \sum_{i=1}^N \int_{\mathcal{W}_i} |c_{\mathcal{W}_i} - x|^2 \diff\rho^*.
    \end{equation*}
    We identify the second term as the intrinsic error $u(t)$, which is bounded because the domain of $\rho^*$ is bounded. 
    Then, we can write the Lyapunov decay~\eqref{eq:sdot-flow-decay} in terms of the SD-OT functional and the error.
    \begin{equation*}
      \frac{d}{dt} W_2^2(\rho^N, \rho^*) = -2 W_2^2(\rho^N, \rho^*) + 2u(t).
    \end{equation*}
  The ultimate dISS bound can be obtained as:
  \begin{equation*}
    \limsup_{t\to\infty} W_2^2(\rho^N, \rho^*)  \leq \frac{1}{2}\limsup_{t\to\infty} (2u(t))\leq C N^{-2/d}.
  \end{equation*}
  Where we used that for $N$ particles optimally placed in $\rho^*$ the error term is bounded by 
  $C_d D^2 N^{-2/d}$, where $D = \text{diam}(\text{supp}(\rho^*))$ and $C_d$ 
  is a constant depending on the dimension~\cite{PZ:82}, and that for sufficiently large $N$, the 
  algorithm assymtoticaly approaches the global minima.
\end{proof}

\section{Numerical Examples}
\subsection{KDE Sinkhorn algorithm}

To illustrate the results, we implement a KDE-based swarm control
using the regularized Wasserstein distance.  The velocity field
$\nabla\phi_u$ is obtained via a convolutional Sinkhorn
algorithm~\cite{JS:15} in a centralized manner, after which the
control laws for each agent follow~\eqref{eq:input-kernel}.
 % The results
% are shown in Figure~\ref{fig:kde-plots} where the dependence on the
% final distance is plotted against different regularization values and
% different number of agents. 
The results in Figures~\ref{fig:kde-plots}(a), (b), and (c) confirm how the
distance of the steady state distribution to a target distribution
increase with the strength of the disturbance, while
Figure~\ref{fig:kde-plots} (d) shows how it decreases as the
number of agents increases. 
\begin{figure}[h!]
\centering

% --- Top Row ---
\begin{subfigure}{0.48\columnwidth}
    \hspace*{0.35cm}
    \includegraphics[width=\linewidth]{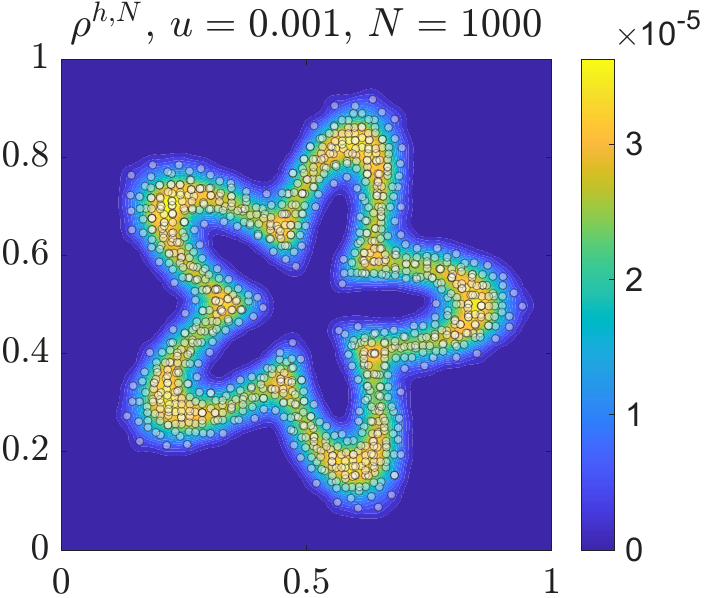}
    \caption{}
\end{subfigure}
\hfill % Dynamically fills space between figures
\begin{subfigure}{0.48\columnwidth}
    \hspace*{0.2cm}
    \includegraphics[width=\linewidth]{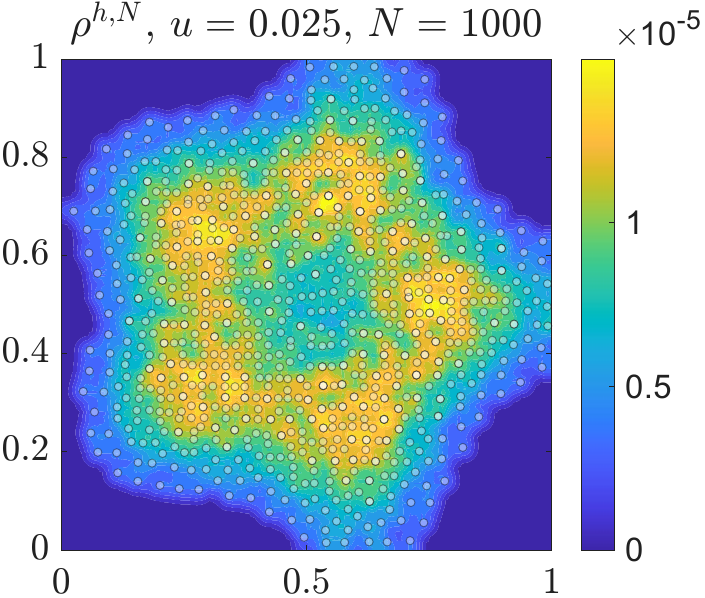}
    \caption{}
\end{subfigure}

\vspace{0.5em}

% --- Bottom Row ---
\begin{subfigure}{0.48\columnwidth}
    \centering
    \includegraphics[width=\linewidth]{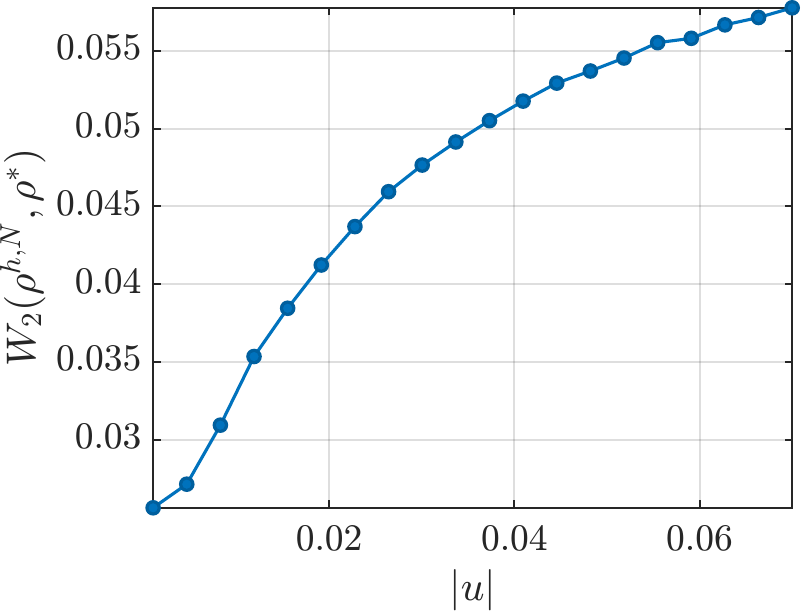}
    \caption{}
\end{subfigure}
\hfill % Matches the dynamic spacing of the top row perfectly
\begin{subfigure}{0.48\columnwidth}
    \centering
    \includegraphics[width=\linewidth]{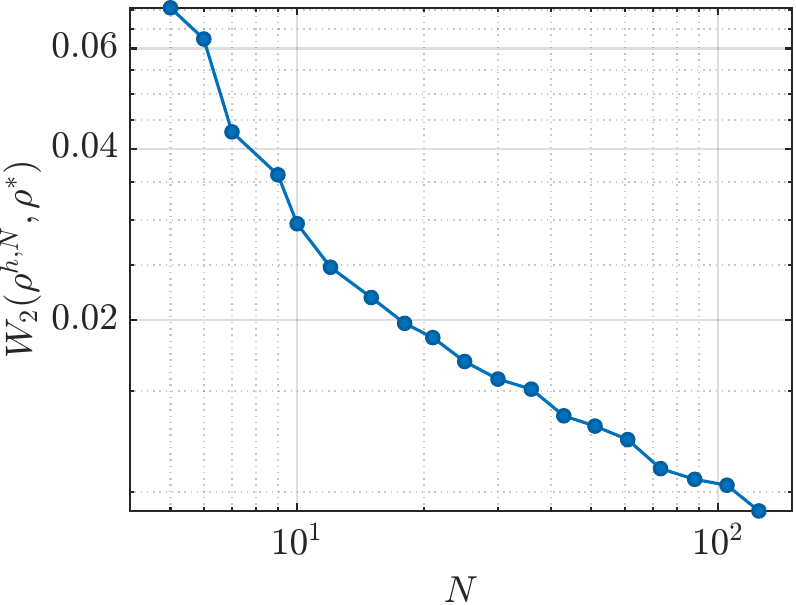}
    \caption{}
\end{subfigure}

\vspace{0.5em}
% --- Bottom Row: Single Side-by-Side Evolution Figure ---
\begin{subfigure}{0.48\columnwidth}
    \includegraphics[width=\linewidth]{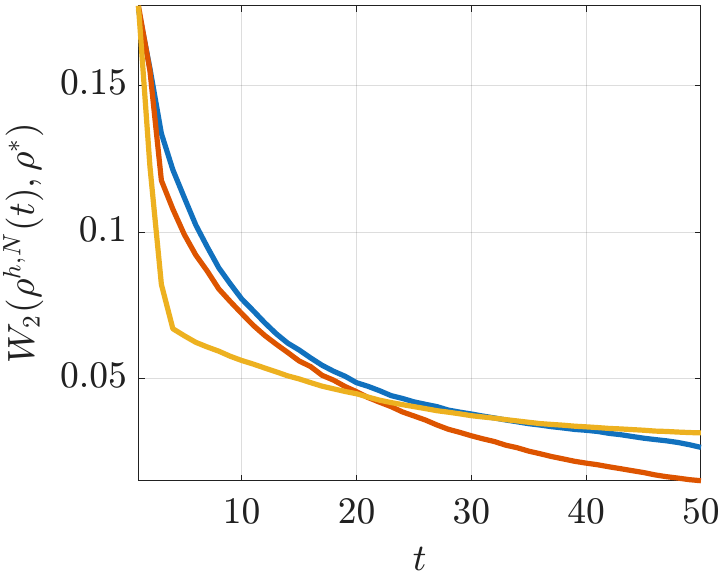}
    \caption{}
\end{subfigure}
\hfill % Matches the dynamic spacing of the top row perfectly
\begin{subfigure}{0.48\columnwidth}
    \centering
    \includegraphics[width=\linewidth]{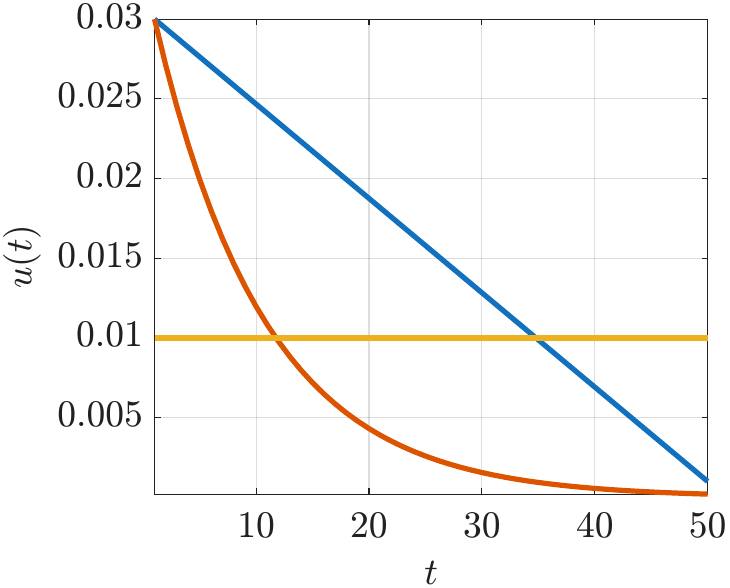}
    \caption{}
\end{subfigure}

\caption{\small \textbf{(a), (b).} Final agent positions for the regularized
  KDE flow with $N=1000$ agents and $u=1\cdot 10^{-3}$, $u=0.025$
  respectively. \textbf{(c), (d).} Final $W_2(\rho^{h,N},\rho^*)$
  distance with varying regularizing parameter $u$, and varying number
  of agents $N$.\textbf{(e), (f).} Evolution of the $W_2(\rho^{h,N},\rho^*)$ distance
  for three different signals $u(t)$. }
\label{fig:kde-plots}
\end{figure}
\section{Conclusions}
This paper introduces a framework to analyse disturbances in
large-scale multi-agent systems, via the notion of distributional
ISS. With this, we are able to study the robustness of Wasserstein
gradient flows of certain relevant functionals, such as the
Wasserstein distance to a target distribution. Moreover, we
illustrated the notion with different examples motivated by swarm
control, such as entropic disturbances and sample approximations, and
we provided simulation results confirming the bounds. Current work
includes the expansion of the scope of this work to other discrete
algorithms and dynamic systems in probability spaces.

\section*{ACKNOWLEDGEMENTS}
\textit{Gemini Pro was used to assist with manuscript graphics,
code troubleshooting, and literature search.}

\bibliographystyle{unsrt}
\bibliography{alias,SMD-add,SM,JC}
\begin{appendices}
\section{Proof of Lemma~\ref{le:convex}}\label{sec:app-le} We
    proof the result for an absolutely continuous $\rho^*$ (minimizer
    of $F(\rho)$) and $\rho_0$.  To obtain quadratic growth we
    evaluate the convexity condition at the minimizer, then the
    gradient of the first variation evaluated at the minimizer
    vanishes and we
    obtain the desired inequality.\\
    For gradient dominance, we have:
\begin{equation*}
F(\rho_0)-F(\rho^*)\leq -\langle\nabla\frac{\delta F}{\delta \rho}|_{\rho_0},v\rangle_{\rho_0}
-\frac{\lambda}{2} W_2^2(\rho_0,\rho^*).
\end{equation*}
By Cauchy Schwarz:
\begin{equation*}
  -\langle\nabla \frac{\delta F}{\delta \rho}|_{\rho_0},v\rangle_{\rho_0}\leq \|\nabla \frac{\delta F}{\delta \rho}|_{\rho_0}\|_{\rho_0}\|v\|_{\rho_0},
\end{equation*}
As $v=T_{\rho_0\to\rho^*}(x)-x$, its weighted norm reduces to the Wasserstein distance. Then:
\begin{equation*}
  F(\rho_0)-F(\rho^*)\leq \|\nabla \frac{\delta F}{\delta \rho}|_{\rho_0}\|_{\rho_0}W_2(\rho_0,\rho^*)-\frac{\lambda}{2} W_2^2(\rho_0,\rho^*).
\end{equation*}
The right hand side is a quadratic expression that is upper bounded as follows:
\begin{equation*}
  \|\nabla \frac{\delta F}{\delta \rho}|_{\rho_0}\|_{\rho_0}W_2(\rho_0,\rho^*)
  -\frac{\lambda}{2} W_2^2(\rho_0,\rho^*)\leq   \|\nabla \frac{\delta F}{\delta \rho}|_{\rho_0}\|_{\rho_0}^2 \frac{1}{2\lambda}.
\end{equation*}
Rearranging yields gradient dominance.

\section{Proof of Theorem~\ref{thm:diss-lyap}}\label{sec:app-proof-lyap}

Take a control input $u_t$ and define the set:
  \begin{equation*}
    I = \left\{ \rho \in \Ptwo\,|\,\ V(\rho) \leq \psi_2 \circ \chi(\|u\|_c) \right\},
  \end{equation*}
   Assume that $\rho_0\in I$,
  but at some time $t$, $\rho_t\notin I$.  As solutions
  to~\eqref{eq:cont-flow} are absolutely continuous, they are also
  continuous wrt the $W_2$ metric and there exists a time $s$ such
  that $\rho_{s}$ must lie on the boundary of $I$, meaning
  $V(\rho_{s}) = \psi_2 \circ \chi(\|u\|_c)$. In that case, the
  $\mathcal{K}_\infty$ property of $\psi_2$ implies
  $W_2(\rho_s, \mathfrak{P}^*) \geq \chi(\|u\|_c)$. Now consider the
  shifted input $\tilde{u}(\tau) = u(\tau + s)$ that drives the system
  for $\tau \geq 0$. The fact $\|\tilde{u}\|_c \leq \|u\|_c$ ensures
  $W_2(\rho_s, \mathfrak{P}^*) \geq \chi(\|\tilde{u}\|_c)$, which
  activates the decay condition
  $\dot{V}_{\tilde{u}}(\rho_s) \leq -\alpha(W_2(\rho_s,
  \mathfrak{P}^*)) < 0$. This strict decrease implies that $V(\rho_t)$
  must immediately decline for $t > s$, contradicting the assumption
  that $s$ is the first escape time (as $\dot{V}<0$, it cannot escape
  the set).

  If $\rho_0\notin I$, the trajectory satisfies the differential
  inequality:
  \begin{equation*}
    \dot{V}(\rho_t) \leq -\alpha(W_2(\rho_t, \mathfrak{P}^*))
    \leq -\alpha \circ \psi_2^{-1}(V(\rho_t)),
  \end{equation*}
  where the second inequality follows from the upper bound
  $V(\rho) \leq \psi_2(W_2(\rho, \mathfrak{P}^*))$.  From the
  comparison principle for $V$, which applies for real-valued
  functionals, there exists $\tilde{\beta} \in \mathcal{KL}$ such
  that:
  \begin{equation*}
    V(\rho_t) \leq \tilde{\beta}(V(\rho_0), t), \quad t \geq 0.
  \end{equation*}
  Consequently, the Wasserstein distance decays according to:
  \begin{align} \label{eq:kl-bound} W_2(\rho_t, \mathfrak{P}^*) \leq
    \beta(W_2(\rho_0, \mathfrak{P}^*), t), \quad \forall t \text{ s.t. } \rho_t \notin I,
  \end{align}
  where
  $\beta(r, t) = \psi_1^{-1} \circ \tilde{\beta}(\psi_2(r), t)$
  for all $r, t \geq 0$. From the properties of $\mathcal{KL}$
  functions, and continuity of $\rho_t$ wrt~$W_2$, 
  \begin{equation*}
    t_1 := \inf_{t \geq 0} \{\rho_t =\rho(t,\rho_0,u) \in I\} < \infty.
  \end{equation*}By the
  invariance of $I$, for all $t > t_1$ we have:
  \begin{equation} \label{eq:iss-bound}
    W_2(\rho_t, \mathfrak{P}^*) \leq \gamma(\|u\|_c),
  \end{equation}
  where
  $\gamma = \psi_1^{-1} \circ \psi_2 \circ \chi \in
  \mathcal{K}$. Combining \eqref{eq:kl-bound} and \eqref{eq:iss-bound}
  yields the distributional ISS estimate:
  \begin{equation*}
    W_2(\rho_t, \mathfrak{P}^*) \leq \max\left\{ \beta(W_2(\rho_0, \mathfrak{P}^*), t),  \gamma(\|u\|_c) \right\}, \quad t \geq 0.
  \end{equation*}

\section{Proof of Proposition\ref{prop:nss-diss}}\label{sec:app-proof-nss}

  \textbf{($\text{dISS} \implies \text{NSS}$):} We substitute the
  indentity~\eqref{eq:exp-w2} into the dISS condition, which
  simplifies to
  \begin{equation*}
    \mathbb{E}_{x_t\sim\rho_t}[\dist^2(x_t,M)]\leq (\beta(\dist(x_0,M),t)+\gamma(\|u\|_c))^2,
  \end{equation*}
  where we take $\rho_0=\delta_{x_0}$. Applying Markov's
  inequality:
  \begin{align*}
    P\{\dist(x_t,M)\geq R\}=P(\dist^2(x_t,M)\geq R^2)
    \leq \\\frac{\mathbb{E}[\dist^2(x_t,M)]}{R^2}\leq\frac{(\beta(\dist(x_0,M),t)+\gamma(\|u\|_c))^2}{R^2}.
  \end{align*}
  Using the value
  $R=\frac{1}{\sqrt{\epsilon}}(\beta(\dist(x_0,M),t)+\gamma(\|u\|_c))$,
  for any value $\epsilon\in(0,1)$ gives the desired result:
  \begin{equation*}
    P\{\dist(x_t,M)\geq \beta_\epsilon(\dist(x_0,M),t)+\gamma_\epsilon(\|u\|_c)\}\leq \epsilon,
  \end{equation*}
  where $\gamma_\epsilon=\gamma/\sqrt{\epsilon}$,
    $\beta_\epsilon =\beta/\sqrt{\epsilon}$.  Using the size function
  $w_M(x)\coloneqq\dist(x,M)$, we recover the NSS formulation. 

  \textbf{($\text{NSS} \implies \text{dISS}$):} In this case, for
  simplicity, we reduce the target to the origin, which in
  distributional space is a Dirac delta $\PM=\{\delta_0\}$, and we
  take as size function $w_M(x)=|x|^2$.  The result can be extended to
  more general sets and size functions, in a similar way as for the
  ISS case.  First, apply the NSS condition for the particular
  case of no covariance disturbance. As the equation is
  deterministic, one obtains a stability condition for the
  unperturbed flow, that is:
\begin{equation*}
  \mathbb{P}\{|x_t|^2\leq \beta(|x_0|^2,t)\}=1.
\end{equation*}
That is, we have a global $\mathcal{KL}$ decay of the state. Under
common regularity conditions, one can apply a converse Lyapunov
theorem for the deterministic part of the SDE
in~\eqref{eq:flow-stochastic}, which guarantees the existence of a
smooth, strict Lyapunov function $V: \Omega \to \mathbb{R}_{\geq 0}$
and class $\mathcal{K}_\infty$ functions
$\underline\alpha, \overline{\alpha}, \eta$ such that for all $x \in \Omega$,
\begin{align}
    \underline\alpha{}(|x|^2) &\leq V(x) \leq \overline{\alpha}(|x|^2), \label{eq:conv_bound1} \\
    \nabla V(x) \cdot f(x) &\leq -\eta(|x|^2). \label{eq:conv_bound2}
\end{align}
We consider as our dISS Lyapunov functional $\mathcal{V}(\rho)$ the
expected value of $V$ under the distribution $\rho$,
\begin{equation*}
  \mathcal{V}(\rho) = \int_\Omega V(x) \diff\rho.
\end{equation*}
In the following, we assume the appropriate convexity of the comparison functions
$\underline\alpha,\overline{\alpha}, \eta$ so that we can apply Jensen's inequality when needed.
In~\cite{DMN-JC:14}, this is formalized by assuming that the Lyapunov function $V(x)$
and the squared norm $|x|^2$ satify a certain comparison inequality. 
The condition is mild in compact domains, as seen in Remark~\ref{re:compact}
%When $\Omega$ is compact
%these assumptions become mild, as we only require local quadratic growth near the
%origin \footnote{Note that if the slope at the origin does not blow up, we can always
%bound the comparison functions on a bounded domain using linear functions
%(e.g., defining a linear upper bound as $\tilde{\alpha}(s) = s \sup_{s \in (0, R^2]} \frac{\alpha(s)}{s}$,
%where $R=\text{diam}(\Omega)$). 
%Because linear functions are simultaneously convex and concave, Jensen's inequality applies seamlessly.}.
In this way, integrating~\eqref{eq:conv_bound1} with respect to $\rho$, we obtain
the positivity bounds of the dISS Lyapunov 

\begin{equation*}
  \underline\alpha\left(\int_\Omega |x|^2 \diff\rho\right)
  \leq \mathcal{V}(\rho) \leq \overline{\alpha}\left(\int_\Omega |x|^2 \diff\rho\right).
\end{equation*}
Next, we evaluate the decay of the functional along the perturbed flow
given by the PDE in~\eqref{eq:flow-stochastic}.  Using the definition
of the time derivative of a functional along the flow, we have:
\begin{align*}
  \dot{\mathcal{V}}(\rho) &= \int_\Omega
                            \nabla V(x) \cdot \left(f(x) - \frac{1}{2\rho}\nabla\cdot(Q\rho)\right) \diff\rho \\
                          & = \int_\Omega
                            \nabla V(x) \cdot f(x) \diff\rho - \frac{1}{2} \int_\Omega \nabla V(x) \nabla\cdot(Q\rho) \diff x.
\end{align*}
We now bound both terms. For the drift term,
substituting~\eqref{eq:conv_bound2} and applying Jensen's inequality
yields:
\begin{equation*}
  \int_\Omega \nabla V(x) \cdot f(x) \diff\rho
  \leq \int_\Omega -\eta(|x|^2) \diff\rho \leq -\eta\left(W_2^2(\rho, \delta_0)\right).
\end{equation*}
For the diffusive term, applying integration by parts,
{\small
\begin{equation*}
  \frac{1}{2}
  \int_\Omega \nabla V(x) \nabla\cdot(Q\rho) \diff x =
  -\frac{1}{2} \int_\Omega \text{Tr}\left( g(x)^\top \nabla^2 V(x) g(x) u_t \right) \diff\rho,
\end{equation*}
}
Since $\Omega$ is compact, $V$ is smooth, and $g$ is continuous, the operator norm is globally
bounded on the domain. There exists a constant $L_g > 0$ such that
$\|\frac{1}{2} g(x)^\top \nabla^2 V(x) g(x)\| \leq L_g$ for all $x \in \Omega$. Thus:
{\small
\begin{equation*}
  \frac{1}{2} \int_\Omega \text{Tr}\left( g(x)^\top
    \nabla^2 V(x) g(x) u_t \right) \diff\rho \leq L_g \|u_t\| \int_\Omega \diff\rho = L_g \|u_t\|.
\end{equation*}
}
Defining $\chi(s) = \eta(\sqrt{s})$ and $\gamma(s) = L_g s$,
we recover the sufficient condition for the dISS Lyapunov functional:
\begin{equation*}
    \dot{\mathcal{V}}(\rho) \leq -\chi(W_2(\rho, \delta_0)) + \gamma(\|u_t\|).
\end{equation*}

\section{Proof of Lemma~\ref{le:proper-loss}}\label{sec:app-proof-proper}

Before presenting the proof, we recall first the definition of a proper
loss function.

\begin{definition}[Proper loss function]
  A continuously differentiable function
  $V: \Omega \rightarrow \mathbb{R}$ is said to be a proper loss
  function with respect a compact subset $\mathcal{A} \subset \Omega$
  if it achieves a strict global minimum $V^*$ on $\mathcal{A}$ (where
  $\min_{x \in \Omega} V(x) = V^* = V(y)$, for some $y \in A$)
  and the following holds:\\
  \noindent \textbf{P1)} The function $V - V^*$
  is a proper size function for $(\Omega, \mathcal{A})$.\\
  \noindent \textbf{P2)} The gradient magnitude $|\nabla V|$ is a
   size function for $(\Omega, \mathcal{A})$.  Equivalently, there
   exists a class $\mathcal{K}_{\infty}$ function $\alpha_2$ such that
   $\alpha_2(V(x) - V^*) \leq |\nabla V(x)|^2$ for all $x \in
   \Omega$.\\
   %When this holds with a linear function (i.e.,
   %$c(V(x) - V^*) \leq |\nabla V(x)|^2$ for a constant $c > 0$), we
   %say that $V$ satisfies the Polyak-\L{}ojasiewicz (PL)
   %condition.
   \noindent \textbf{P3)} The gradient $\nabla V$ is locally
   Lipschitz continuous.\oprocend\\
\end{definition}
Note that the minimizers of the functional $\mathcal{V}(\mu)$ are the set
of measures with support on the minimizers of $V(x)$, that is, $\PM$.
We then define $\mathcal{V}^*\coloneq \mathcal{V}(\PM)$, note that
$\mathcal{V}^*=V^*$.
\textit{Quadratic growth:} Since $V$ is a proper loss function,
$V - V^*$ is a size function with respect to $M$.  From the
properties of size functions, $V$ admits a lower bound in terms of
the distance to $M$. We
have $V(x) - V^* \ge \alpha_1(\dist(x, M)^2)$, where $\alpha_1$ is convex
by assumption. Integrating w.r.t.~$\mu$ and using Jensen's inequality:
\begin{align*}
  \int (V(x) - V^*)
  \diff \mu(x) &\ge
                 \alpha_1(\int \dist(x,M)^2 \diff \mu) \\
  \Longleftrightarrow \;
  \mathcal{V}(\mu) - \mathcal{V}^* &\ge \alpha_1(W_2^2(\mu, \PM)).
\end{align*}

\textit{Gradient dominance:} Using 
$|\nabla V(x)|^2 \ge \alpha_2(V(x) - V^*)$, and $\alpha_2$ convexity, 
we can integrate both sides with
respect to $\mu$ and apply Jensen's:
\begin{align*}
  \int |\nabla V(x)|^2
  \diff \mu(x) &\ge \int \alpha_2 (V(x) - V^*) \diff \mu \quad \Longleftrightarrow \\
  \Big\|\nabla
  \frac{\delta \mathcal{V}(\mu)}{\delta \rho}\Big\|_\mu^2
               &\ge \alpha_2 (\mathcal{V}(\mu) - \mathcal{V}^*).
\end{align*}

\textit{Local $l$-smoothness}: Since $V$ is a proper loss function,
its gradient $\nabla V$ is locally Lipschitz.  Restricting our
attention to any compact sublevel set (or bounded region), there
exists a constant $L > 0$ such that for all $x,y$ in this set:
\begin{equation*}
  V(y) \le V(x) + \nabla V(x)^\top(y-x) + \frac{L}{2}|y-x|^2.
\end{equation*}
Let $\pi \in \Pi(\mu, \nu)$ be the optimal transport plan between
$\mu$ and $\nu$.  Integrating the inequality with respect to
$\pi(x,y)$ yields:
\begin{align*}
  \int V(y) \diff \pi
  &\le \int V(x) \diff \pi + \int \nabla V(x)^\top (y-x) \diff \pi \\
                      &\quad + \frac{L}{2} \int |y-x|^2 \diff \pi.
\end{align*}
Assuming that the optimal plan is induced by a deterministic map
$y = T_{\mu\to\nu}(x)$, and due to the optimality of $\pi$, the
inequality evaluates to:
\begin{equation*}
  \mathcal{V}(\nu) \le \mathcal{V}(\mu)
  +
  \langle\nabla\frac{\delta \mathcal{V}(\mu)}{\delta \rho},T_{\mu\to\nu}-\Id \rangle_\mu
  + \frac{L}{2} W_2^2(\mu, \nu).
\end{equation*}
This confirms $\mathcal{V}$ inherits the $l$-smoothness of $V$
along optimal transport geodesics.

\end{appendices}

\end{document}

%% file: macros-abbrev.tex
\newcommand{\real}{\ensuremath{\mathbb{R}}}

\newcommand{\realnonnegative}{\ensuremath{\mathbb{R}}_{\ge 0}}

\newcommand{\supscr}[2]{#1^{\textup{#2}}}
\newcommand{\setdef}[2]{\{#1 \; | \; #2\}}

\newcommand{\dist}{\operatorname{dist}}

\renewcommand{\epsilon}{\varepsilon}

\usepackage{psfrag}